%% file: mainQBN.tex
\documentclass[a4paper,UKenglish,cleveref,thm-restate,pdfa]{lipics-v2021}
\nolinenumbers
\pdfoutput=1 
\hideLIPIcs 

\input{./macrosBPN} 

\title{Quantum Bayesian Networks:\texorpdfstring{\\}{ }Compositionality and Typing via Linear Logic}

\titlerunning{Quantum Bayesian Networks: Compositionality and Typing via Linear Logic}

\author{R{\'e}mi {Di Guardia}}{Université Paris Cité, CNRS, Inria, IRIF, F-75013, Paris, France \and \url{https://remidig.github.io/}}{remi.di.guardia@ens-lyon.org}{https://orcid.org/0009-0004-8632-108X}{}
\author{Thomas Ehrhard}{Université Paris Cité, CNRS, IRIF, F-75013, Paris, France}{ehrhard@irif.fr}{https://orcid.org/0000-0001-5231-5504}{}
\author{Claudia Faggian}{Université Paris Cité, CNRS, IRIF, F-75013, Paris, France}{faggian@irif.fr}{}{}

\authorrunning{R. {Di Guardia}, T. Ehrhard, and C. Faggian}

\Copyright{R{\'e}mi {Di Guardia}, Thomas Ehrhard, and Claudia Faggian}



\relatedversiondetails[cite=QBPN_fscd]{Conference version}{https://doi.org/10.4230/LIPIcs.FSCD.2026.16}

\funding{by the Plan France 2030 through the PEPR project EPiQ (ANR-22-PETQ0007) and by the European Union through the MSCA SE project QCOMICAL (Grant Agreement ID: 101182520).}

\acknowledgements{We are grateful to Benoît Valiron for insightful discussions and to the anonymous referees for their valuable comments and the suggested references.}

\EventEditors{Frank Pfenning}
\EventNoEds{1}
\EventLongTitle{11th International Conference on Formal Structures for Computation and Deduction (FSCD 2026)}
\EventShortTitle{FSCD 2026}
\EventAcronym{FSCD}
\EventYear{2026}
\EventDate{July 20--23, 2026}
\EventLocation{Lisbon, Portugal}
\EventLogo{}
\SeriesVolume{378}
\ArticleNo{16}

\begin{document}
	
\maketitle
	
\begin{abstract}
Quantum Bayesian networks~\cite{GBN2014} provide a mathematical formalism to describe causal relations, to analyse correlations, and to predict the probabilities of measurement outcomes, in systems involving both \emph{classical and quantum} data.
They generalize Pearl's Bayesian networks~\cite{Pearl_book}---prominent graphical models for classical probabilistic reasoning and inference.

The goal of this paper is to bring compositional principles and a typing discipline into this setting. 
A key feature of our compositional semantics is that when all causes are classical, it coincides with the standard factor-based semantics of Bayesian networks, while in the purely quantum case it reduces to tensor networks.
We then propose a typed formalism based on linear logic proof-nets, where types ensure well-behaved composition of systems, and which we prove sound and complete with respect to quantum Bayesian networks.
\end{abstract}

\input{00_Introduction}
\input{01_Background}
\input{02_QuantumDistributions}
\input{03_QBNs}
\input{04_QPNs}
\input{05_Conclusions}

\bibliography{biblio_extracted}

\newpage

\input{99_Appendix}

\end{document}

%% file: macrosBPN.tex
\usepackage
{todonotes}

\usepackage{graphicx}
\usepackage{braket}
\usepackage{tikz}
\usetikzlibrary{arrows, decorations.markings, decorations.pathmorphing, decorations.pathreplacing, shapes.misc, shapes.geometric, arrows.meta, fit, patterns}
\usetikzlibrary{calc}

\usepackage{amssymb}
\usepackage{stmaryrd}
\usepackage{proof}

\usepackage{mathrsfs}

\usepackage{mathtools}
\usepackage{ifthen}

\usepackage{cmll}
\usepackage{xspace}
\usepackage{adjustbox}
\adjustboxset{max width = \linewidth, center = \linewidth}
\usepackage{ebproof} 
\ebproofset{
	template={$\vdash \inserttext$},
	right label template={$\scriptstyle(\inserttext)$},
	left label template={$\scriptstyle\inserttext$},
	center=false,
}
\usepackage{url}

\NewDocumentCommand{\ntodo}{O{}mm}{\todo[#1]{#2: #3}}
\NewDocumentCommand{\remi}{O{}m}{\ntodo[color=olive!50,#1]{R}{#2}}

\Crefname{equation}{Eq.}{Eqs.}
\Crefname{section}{Sec.}{Sections}
\Crefname{figure}{Fig.}{Figs.}
\Crefname{theorem}{Thm.}{Thm.}
\Crefname{proposition}{Prop.}{Prop.}
\Crefname{remark}{Rem.}{Rem.}
\Crefname{corollary}{Cor.}{Cor.}
\Crefname{definition}{Def.}{Def.}
\Crefname{example}{Ex.}{Ex.}
\Crefname{property}{Property}{Properties}

\adjustboxset{max width = \linewidth, center = \linewidth}

\NewDocumentEnvironment{genscope}{O{black}O{solid,->}O{0.8pt}}{
\begin{scope}[
every node/.style={circle,minimum size=8mm,inner sep=0,draw=#1,line width=#3},
every edge/.style={draw=#1,#2,line width=#3},
every path/.style={thick,draw=#1,#2,-,line width=#3}]
\color{#1}
}{
\end{scope}
}

\tikzset{
    named edge/.style={draw=none,inner sep=1pt,outer sep=2pt,rectangle,minimum size=0pt,auto},
    named path/.style={named edge,sloped},
}

\newcommand*{\resp}{resp.}
\newcommand*{\ie}{\emph{i.e.}\xspace}
\newcommand*{\eg}{\emph{e.g.}\xspace}
\newcommand*{\aka}{\emph{a.k.a.}}
\newcommand*{\wrt}{\emph{w.r.t.}}

\newcommand{\textdef}[1]{\textbf{#1}}
\newcommand*{\definitive}[1]{\textbf{#1}}

\newcommand*{\defeq}{\stackrel{\text{def}}{=}}

\newcommand{\der}{\vdash}
\newcommand*{\symdif}{\mathrel{\Delta}}
\newcommand*{\difset}{\mathrel{-}}

\newcommand{\true}{\mathtt{t}}
\newcommand{\false}{\mathtt{f}}
\newcommand{\Real}{\mathbb{R}}

\newcommand*{\orth}{^\perp}
\newcommand*{\tens}{\otimes}
\newcommand*{\lollipop}{\multimap}

\newcommand{\FormA}{F} 
\newcommand{\FormB}{G} 

\newcommand{\MLL}{\ensuremath{\mathsf{MLL}}\xspace}

\newcommand*{\posatom}[1]{{#1}^+}
\newcommand*{\negatom}[1]{{#1}^-}

\newcommand*{\rulenametext}[1]{\textnormal{\textit{#1}}}
\newcommand*{\ax}{\rulenametext{ax}}
\newcommand*{\cut}{\rulenametext{cut}}

\newcommand{\X}{X^+}

\newcommand{\Nm}[1]{\mathsf{Nm}(#1)}

\newcommand*{\boxof}[1]{\ensuremath{\bbox^{#1}}}

\newcommand*{\logsys}[1]{\textnormal{\textsf{#1}}}
\renewcommand*{\MLL}{\ensuremath{\logsys{MLL}}\xspace}

\newcommand*{\somepn}{\mathcal{R}}
\newcommand*{\somenode}{n}
\newcommand*{\G}{\mathcal{G}}

\newcommand{\netN}{\mathcal N}
\newcommand{\N}{\netN}
\newcommand{\Net}{\N}
\newcommand{\netR}{\mathcal R}
\newcommand{\R}{\netR}

\newcommand{\netE}{\mathcal E}
\newcommand{\E}{\netE}

\newcommand{\bbox}{\mathtt{b}}

\newcommand{\Boxes}[1]{\mathsf{Boxes}(#1)}

\newcommand{\pol}[1]{\mathit{Pol}({#1})}

\newcommand{\Red}[1][]{\xrightarrow{#1}}

\newcommand{\den}[1]{\llbracket{#1}\rrbracket}
\newcommand{\sem}[1]{\llparenthesis{#1}\rrparenthesis}
\newcommand*{\compfact}{\odot}

\newcommand{\nnReal}{\Real^+_0}
\newcommand*{\setofrv}[1]{\mathsf{Val}(#1)}

\newcommand{\Val}[1]{\mathsf{Val}(#1)}

\newcommand{\ftone}{\mathbf{1}}

\newcommand*{\somebayes}{\mathcal{B}}
\newcommand*{\bn}{\somebayes}

\newcommand{\bX}{\mathbf X}
\newcommand{\bY}{\mathbf Y}

\newcommand{\bx}{{\mathbf x}}
\newcommand{\by}{{\mathbf y}}

\newcommand{\be}{{\mathbf e}}

\newcommand{\tr}[1]{\mathrm{tr}\left(#1\right)}
\newcommand{\ptr}[1]{\mathrm{tr}_{#1}}

\renewcommand{\L}{\mathcal{L}}
\newcommand{\Pos}[1]{{\L(#1)}^+}
\newcommand{\D}{\mathcal{D}}

\newcommand{\Eop}{\mathcal E}

\newcommand*{\baseof}[1]{\textsf{Basis}_{#1}}

\newcommand*{\C}{\mathbb{C}}

\newcommand{\HH}{\mathcal{H}}

\newcommand{\qfactfunc}[1]{\widehat{#1}}

\newcommand{\M}{\rho}

\newcommand{\reg}{Q}
\newcommand{\qreg}{\reg}

\newcommand{\bQ}{\mathbf Q}
\newcommand*{\someqfact}{\phi}
\newcommand*{\someqfactset}{\Phi}

\newcommand*{\id}[1][]{\mathrm{Id}_{#1}}
\newcommand*{\qcup}[1][]{\mathrm{cup}_{#1}}
\newcommand*{\qcap}[1][]{\mathrm{cap}_{#1}}
\newcommand*{\cc}[1][]{\mathrm{c}_{#1}}
\newcommand*{\cw}[1][]{\mathrm{w}_{#1}}

\newcommand{\idop}{\id}

\newcommand*{\choiiso}{\mathcal{J}}

\newcommand{\rv}{r.v.\xspace}
\newcommand{\rvs}{r.v.s\xspace}

\newcommand*{\qfactor}{Q-factor}
\newcommand*{\Qfactor}{\qfactor}

\newcommand*{\QCPT}{Quantum CPT}

\newcommand*{\qCPT}{Q-\textsf{cpt}\xspace}
\newcommand*{\qCPTs}{Q-\textsf{cpt}s\xspace}

\newcommand*{\QBN}{quantum Bayesian network}
\newcommand*{\QBNs}{Quantum Bayesian networks\xspace}

\newcommand*{\qqbn}{\textsf{QBN}}

\newcommand*{\qpn}{quantum proof-net\xspace}
\newcommand*{\qpns}{quantum proof-nets\xspace}

\newcommand*{\Qpns}{Quantum proof-nets\xspace}

\newcommand*{\qbpn}{\textsf{qpn}\xspace}
\newcommand*{\qbpns}{\textsf{qpns}\xspace}

\newcommand*{\main}{output\xspace}

\newcommand*{\good}{closed\xspace}

\newcommand{\cpt}{CPT\xspace}

\newcommand{\pn}{proof-net\xspace}
\newcommand{\pns}{proof-nets\xspace}

\newcommand{\pts}{proof-trees\xspace}

%% file: 00_Introduction.tex
\section{Introduction}

Pearl's \emph{Bayesian networks}~\cite{Pearl86,Pearl_book} provide a framework for reasoning under conditions of uncertainty and partial knowledge, with a wide range of applications from statistics to epidemiology, economics and computer science.
Bayesian networks have a \emph{dual nature}, serving both as probabilistic graphical models for classical probabilistic reasoning and inference, and as causal models, precising the connections between observed data and causal relations.
When reasoning on \emph{quantum systems}, the classical framework is not general enough to account for entanglement and the non-local correlations observed in Bell experiments.
The development of quantum causal models (see \eg~\cite{BarrettLorenzOreshkov2019} and references therein) is an active research area across quantum information and the foundations of quantum theory, advancing along various axes, whose motivations span from foundational questions and non-locality, to enabling device-independent cryptographic protocols, to facilitating data-driven discovery.

In this paper, we focus on quantum Bayesian networks, a direct generalization of Pearl's networks introduced in foundational work by Henson, Lal, and Pusey~\cite{GBN2014}.
They provide a mathematical framework to describe causal relations, to analyze correlations, and to predict the probabilities of measurement outcomes, in systems involving both classical and quantum data.
The formalism builds on previous work by Leifer and Spekkens~\cite{LeiferSpekkens2013}, where the perspective is that of \emph{quantum theory as a theory of inference}.
Quantum theory is indeed fundamentally probabilistic at its core, as it is concerned with predicting the probabilities of measurement outcomes on a physical system.
In this sense, the prediction task can be framed as a problem of \emph{probabilistic inference over models involving both classical and quantum data}, as \Cref{ex:bell} illustrates.
Probabilistic inference then offers mathematical and logical tools to comprehend, predict, and control quantum phenomena, essential both to the theoretical understanding and to quantum technologies.

\begin{example}[Alice \& Bob: the Bell set-up]
\label{ex:bell}
The directed acyclic graph in \cref{fig:Bell} describes the well-known set-up for the Bell experiment.
Alice and Bob---who stand in widely separated laboratories---are each able to perform two possible measurements on a qubit (for example, measuring it with respect to two different bases). 
Their colleague Quentin prepares a pair of (possibly entangled) qubits, sending one to Alice and the other to Bob. When Alice receives her qubit $q_1$, she chooses to randomly perform one of the two possible measurements, by flipping a coin $X$.
When Bob receives his qubit $q_2$, he also performs a measurement, by flipping a coin $Y$.
The result of the experiment is
\begin{equation}\label{eq:bell}
	\Pr(a,b\mid x,y) =\Pr(a,b,x,y) /\Pr(x,y)
\end{equation}
\ie\ the probability that the (classical) outcomes of Alice and Bob measurements are respectively $a$ and $b$, given outcomes $x$ for $X$ and $y$ for $Y$.
\end{example}

\begin{figure}
	\begin{minipage}[b]{0.39\linewidth}
		\begin{adjustbox}{}
			\begin{tikzpicture}
				\begin{genscope}
					\node (a) at (-1.5,0) {$A$};
					\node (b) at (1.5,0) {$B$};
					\node (c) at (0,1) {$\qreg$};
					\node (x) at (-3,1) {$X$};
					\node (y) at (3,1) {$Y$};
					
					\path (y) edge (b);
					\path (x) edge (a);
					\path (c) edge node[named path]{$q_2$} (b);
					\path (c) edge node[named path]{$q_1$} (a);
				\end{genscope}
			\end{tikzpicture}
		\end{adjustbox}
		\caption{Bell set-up (from~\cite{GBN2014}).}
		\label{fig:Bell}
	\end{minipage}
	\hfill
	\begin{minipage}[b]{0.59\linewidth}
		\begin{adjustbox}{}
			\begin{tikzpicture}[baseline=(base)]
				\coordinate (base) at (0,0);
				\begin{scope}[every node/.style={rectangle,draw=none}, every path/.style={draw=black}]
					\node[draw=black,minimum width=2.5em] (bQ) at (1,0) {$\qreg$};
					\node[draw=black,minimum width=2.5em] (bX) at (3,0) {$X$};
					\node[draw=black,minimum width=2.5em] (bA) at (6,0) {$A$};
					\node[draw=black,minimum width=2.5em] (bB) at (8,0) {$B$};
					\node[draw=black,minimum width=2.5em] (bY) at (10,0) {$Y$};
					
					\node (cX) at (4.5,-1.5) {$\cut$};
					\node (cQ) at (4,-2.25) {$\cut$};
					\node (pQ) at (6.7,-1.5) {$\parr$};
					\node (cY) at (9.1,-1.5) {$\cut$};
					\path (3.3,-.25) |- node[named edge,pos=.5,above right]{$X^+$} (cX) -| node[named edge,pos=.5,above left]{$X^-$} (5.7,-.25);
					\path (1.3,-.25) |- node[named edge,pos=.5,above right]{$\qreg^+$} (cQ) -| node[named edge,pos=.5,above left]{$\qreg^-$} (pQ);
					\path[out=180,in=-90] (pQ) edge node[named edge,pos=.4,right]{$q_1^-$} (5.9,-.25);
					\path[out=0,in=-90] (pQ) edge node[named edge,pos=.4,left]{$q_2^-$} (7.7,-.25);
					\path (10.3,-.25) |- node[named edge,pos=.5,above left]{$Y^+$} (cY) -| node[named edge,pos=.5,above right]{$Y^-$} (7.9,-.25);
					\path (6.3,-.25) -- node[named edge,pos=.75]{$A^+$} ++(0,-.6);
					\path (8.3,-.25) -- node[named edge,pos=.75]{$B^+$} ++(0,-.6);
				\end{scope}
			\end{tikzpicture}
		\end{adjustbox}
		\caption{Bell set-up as a quantum proof-net.}
		\label{ex:bell_pn}
	\end{minipage}
\end{figure}

\subparagraph{Bayesian Networks.}

In the theory of Bayesian networks, the causal structure is encoded by a \emph{directed acyclic graph} (DAG) where nodes represent random variables and edges express conditional dependencies.
The strength of the dependencies (or the degree of knowledge) is quantified by conditional probability tables.
A key benefit of Bayesian networks is to provide a compact representation of large probability distributions, and efficient inference algorithms (both exact and approximate) to answer queries about the underlying distribution without explicitly constructing it in full.
Critically, the DAG includes both \emph{observable} variables of interest, and \emph{hidden} (unobserved) ones; in figures, we adopt the common convention of denoting classical hidden variables by shaded nodes. In the quantum setting---as in Bell's theorems---hidden variables play a central role, and they may correspond to quantum systems.

\subparagraph{Semantics and Inference.}

The semantics of a Bayesian network is the probability distribution it defines.
More accurately, what one seeks is the \emph{marginal distribution over the variables of interest}.
Exact inference computes it precisely---this involves two key operations:

{\hfill
	product (\emph{composing}) 
	+ summing-out (\emph{hiding}) irrelevant variables.
\hfill}

The formalization and theory of inference rely on a class of functions, known as factors (see \Cref{sec:BN}), which is an abstraction of conditional probability distributions.
Tractability and efficiency rely on their properties, and specifically on two key aspects: the product of factors \emph{inherently shares variables}, and product and sum \emph{distribute} under suitable conditions---pushing the sum on smaller components reduces the size of computations.

\subparagraph{Quantum Bayesian Networks, issues.}

Quantum Bayesian networks are still an emerging field, not as developed as their classical counterparts.
A crucial missing feature is the ability to compute the semantics of the model (the desired marginal distribution) through intermediate, partial computations, without ever computing the full joint distribution.
Put differently, what is lacking is \emph{compositionality}, which would enable computing a model's denotation as a function of its subparts, alongside modular reasoning.
A closely related question concerns \emph{modularity}: when can causal descriptions of systems as subparts be used to construct larger models?
A well-established tool to ensure modularity are \emph{types}, that specify precise contracts (\eg\ input/output behaviors) for the components of a system. 

\subparagraph{The goal of this paper.}

We have two main objectives in this paper.
\begin{itemize}
\item
To address the lack of compositionality and modularity in the setting of quantum Bayesian networks by introducing methods and concepts from denotational semantics and proof theory, thereby enabling compositional principles and a typing discipline.
\item
To develop a framework fully compatible with Bayesian networks and Bayesian inference, thus paving the way for the application of the techniques developed in that context.
\end{itemize}

\emph{Compositionality and modularity} facilitate reasoning about complex systems and their properties, ensuring that the meaning of the system can be derived systematically and in a principled way from the meanings of its parts, and allowing components to be analyzed and replaced independently.
Compositionality also enables modular reasoning about inference.

\emph{Types} serve as abstract characterizations of systems behaviors, constraining and guiding their formation.
Types discipline statically guarantees semantic properties such as termination, consistency, and compositional correctness.
Well-typed programs exhibit well-behaved execution and preserve quantitative (\eg\ probabilistic) invariants throughout their evaluation.

\emph{Encompassing the semantics of Bayesian networks.}
As seen in \cref{ex:bell} (the Bell set-up), even for systems involving quantum sources of causality, only the classical outcomes of measurements can be observed.
Thus, what a model ultimately defines is a probability distribution over classical variables---such as \cref{eq:bell}.
From the inference perspective, it is desirable to have a framework enabling (when relevant) the extensive set of inference techniques and algorithms developed for Bayesian networks.
To this end, our semantics integrates a key notion of inference algorithms, that of \emph{factor}, which we discuss in \cref{sec:BN}.

\subparagraph{Contributions and challenges.}

Our first contribution is to develop a \emph{compositional semantics}, which allows for the interpretation and modular combination of components.
We adapt Selinger's semantics in~\cite{Selinger2004} to take into account the factor-based approach of Bayesian networks.
The main technical challenge is to conciliate two very different behaviors:
\begin{itemize}
\item
\emph{Classical variables share their values, and they do so in an efficient way}: the mathematical setting underlying Bayesian networks integrates this feature in the very definition of factors product, and exploits it to obtain compact representations and efficient calculations.
\item
\emph{Quantum data cannot be shared}: a defining feature of quantum computing is the \emph{No-Cloning Theorem}, implying qubits cannot be duplicated nor broadcast to multiple receivers.
\end{itemize}

We satisfy both requirements by introducing \emph{quantum factors} in \cref{sec:semantic}.
The mathematical developments in that section are our main and most technical results. 
Remarkably, when all causes are classical, our framework exactly coincides with the standard \emph{factor-based} semantics of Bayesian networks, while in the purely quantum case it behaves like tensor networks.
In \cref{sec:qbn}, we rely on quantum factors to redefine quantum Bayesian networks.
Our formalism is equivalent to that in~\cite{GBN2014}, however our semantics enables a compositional interpretation by sub-components, unlike the original definition (see~\cite[page 12]{GBN2014}), as we discuss in \cref{sec:lack}.
A crucial aspect to achieve compositionality is that quantum factors are closed under both product and sum-out (\ie\ marginalization of unobserved variables).
Finally, in \cref{sec:qpns}, we explore a \emph{typed graphical formalism} based on proof-nets of linear logic, where types ensure well-behaved compositions of systems.
Our key result is that the formalism is \emph{sound and complete} \wrt\ quantum Bayesian networks: every quantum Bayesian network can be represented as a proof-net, and every closed proof-net corresponds to a quantum Bayesian network.

\paragraph*{Motivational examples: compositionality and modularity}
\label{sec:motivational}

Two examples can illustrate the desiderata and issues with compositionality and modularity.

\subparagraph{Compositionality.}

Consider the DAG in \Cref{fig:compositional}.
The nodes $X$ and $Y$ produce a classical output, while the nodes $Q_1$ and $Q_2$ have a quantum nature.
A natural question is:
\emph{Can we compute the semantics of the model in terms of sub-components}---as for example the highlighted sub-graphs? 
The approach in~\cite{GBN2014} does not adapt well to an interpretation by components because it has a global nature, as we will discuss in \Cref{sec:lack}.

\subparagraph{Do parts compose well?}

When considering components (possibly with open inputs\footnote{Bayesian networks with open inputs are called \emph{conditional Bayesian networks}, see \eg~\cite{KollerBook}.}), a natural question somehow dual to the previous one is whether independently defined components compose well.
Consider the three DAGs in \Cref{fig:modular}: $\N_0$ awaits an input $A$ and outputs $C$, $\N_1$ and $\N_2$ both await an input $C$ and output $A$ and $D$.
The graph obtained by plugging together $\N_0$ and $\N_1$ (matching inputs and outputs) is a DAG, while the graph that plugs together $\N_0$ and $\N_2$ has a directed cycle.
By moving to \emph{typed graphs}, we guarantee that composing graphs of compatible types produces a DAG (\cref{sec:PNcompositionality}).
Rather than defining yet-another-syntax, we encode quantum Bayesian networks into the graph syntax of linear logic: proof-nets.
This builds on a recent line of work connecting (classical) Bayesian networks with proof-nets~\cite{EhrhardFP23,popl24,BPN}.
A \pn is typed by a sequent in Multiplicative Linear Logic.
As we will see (\cref{sec:PNcompositionality} and \cref{fig:modularity_pn}), the DAGs in \cref{fig:modular} admit the following typing:
\begin{equation*}
\hfill \N_0\der A\lollipop C \hfill \N_1\der (A\lollipop C)\lollipop D \hfill \N_2\der C \lollipop (A \otimes D) \hfill
\end{equation*}
The DAGs $\N_0$ and $\N_1$ compose together, producing a DAG of output $D$, while $\N_2$ cannot be given any type that matches the one of $\N_0$.

\begin{figure}
\begin{minipage}[b]{0.35\linewidth}
\begin{adjustbox}{max height = 8em, center = \linewidth}
\begin{tikzpicture}[bg_ellipse/.style={ellipse,minimum height=1.8cm,inner sep=0pt}]
		\coordinate (T_Q1) at (0, 0);
		\coordinate (T_Y) at (2.5, 0);
		\coordinate (T_Q2) at (1.5, -1.5);
		\coordinate (T_X) at (-0.5, -3.0);
		
		\node[bg_ellipse, fill=red!10, rotate=80.5, minimum width=4.5cm] at ($(T_Q1)!0.5!(T_X)$) {};
		\node[bg_ellipse, fill=blue!10, rotate=56.3, minimum width=3.2cm] at ($(T_Y)!0.5!(T_Q2)$) {};
\begin{genscope}
		\node (Q1) at (T_Q1) {$\qreg_1$};
		\node[fill=gray!50] (Y) at (T_Y) {$Y$};
		\node (Q2) at (T_Q2) {$\qreg_2$};
		\node (X) at (T_X) {$X$};
		
		\path (Q1) edge (X);
		\path (Q1) edge (Q2);
		\path (Y) edge (Q2);
		\path (Q2) edge (X);
		\path (X) edge ++(0, -0.8);
\end{genscope}
\end{tikzpicture}
\end{adjustbox}
\caption{Compositionality.}
\label{fig:compositional}\label{fig:components}
\end{minipage}
\hfill
\begin{minipage}[b]{0.6\linewidth}
\begin{adjustbox}{max height = 8em, center = \linewidth}
\begin{tikzpicture}[baseline=(base)]
	\node (base) at (-1, 2.5) {$\N_0$};
\begin{genscope}
	\node[fill=gray!50] (B) at (0, 2) {$B$};
	\node (C) at (0, .5) {$C$};
	
	\path (0, 3.25) edge node[named edge,right]{$A$} (B);
	\path (B) edge (C);
	\path (C) edge node[named edge,right]{$C$} (0, -.5);
\end{genscope}
\end{tikzpicture}
\quad\vrule\quad
\begin{tikzpicture}[baseline=(base)]
	\node (base) at (-1.5, 3) {$\N_1$};
\begin{genscope}
	\node[fill=gray!50] (E) at (0, 3) {$E$};
	\node (D) at (-1.5, 1) {$D$};
	\node (A) at (1.5, 2) {$A$};
	
	\path (E) edge (D);
	\path (E) edge (A);
	\path ($(D)+(-0.8, 1.0)$) edge node[named edge,left]{$C$} (D);
	\path (D) edge node[named edge,left]{$D$} ($(D)+(-0.6, -1.0)$);
	\path (A) edge node[named edge,right]{$A$} ($(A)+(0.5, -1.5)$);
\end{genscope}
\end{tikzpicture}
\quad\vrule\quad
\begin{tikzpicture}[baseline=(base)]
	\node (base) at (-1.5, 3) {$\N_2$};
\begin{genscope}
	\node[fill=gray!50] (E) at (0, 3) {$E$};
	\node (D) at (-1.5, 1) {$D$};
	\node (A) at (1.5, 2) {$A$};
	
	\path (E) edge (D);
	\path (E) edge (A);
	\path ($(D)+(-0.8, 1.0)$) edge node[named edge,left]{$C$} (D);
	\path (D) edge node[named edge,left]{$D$} ($(D)+(-0.6, -1.0)$);
	\path (A) edge node[named edge,right]{$A$} ($(A)+(0.5, -1.5)$);
	\path[out=-50, in=-100] (D) edge (A);
\end{genscope}
\end{tikzpicture}
\end{adjustbox}
\caption{Modularity.}
\label{fig:modular}
\end{minipage}
\end{figure}

%% file: 01_Background.tex
\section{Preliminaries}
\label{sec:background}

We recall some basics on Bayesian networks and on quantum computation, respectively referring to~\cite{DarwicheBook} (or~\cite{DarwicheHandbook} for a compact introduction) and to~\cite{Nielsen_book,Watrous_book} for further reading.
We then present quantum Bayesian networks from~\cite{GBN2014}.

\subsection{Classical data and Bayesian Networks}
\label{sec:BN}

Bayesian networks (BNs) are probabilistic graphical models.
Bayesian models provide a formalism for reasoning under conditions of uncertainty or partial knowledge: given a system under study, how \emph{likely} is it that a particular feature is in a particular state?
Every feature of the system is represented by a \emph{random variable}.
For the purpose of modeling, each random variable can be seen as a \emph{name} for an atomic proposition (\eg ``Rain'') which assumes values from a set of states (\eg $\{\true,\false\}$).
The full system is modeled as \emph{a joint probability distribution} over the variables of interest; each element in the sample space represents a possible state.

\subparagraph{Random Variables (\rvs).}

Assuming a countable set of names $X,Y,\dots$, a \emph{random variable} (\rv) is the pair of a name $X$ and a finite set of values $\Val X$.
\emph{W.l.o.g.}, we will assume all random variables to be binary, with $\Val{X}=\{x^\true,x^\false\}$ for a name $X$.
We then silently identify a name $X$ with the \rv $(X, \Val X)$.

We adopt the standard convention of capital letters (\eg $X,Y$) denoting random variables, and lowercase letters (\eg $x,y$) for particular \emph{values} of those variables; $\Pr(x)$ stands for $\Pr(X=x)$.
A finite set of names $\bX=\{X_1, \dots, X_n\}$ defines a ``compound'' \rv where $\Val{\bX}$ is the Cartesian product $\Val{X_1} \times \dots \times \Val{X_n}$; we denote by $\bx$ a tuple in $\Val{\bX}$.

%

\subparagraph{Bayesian Networks.}

A Bayesian network over a set of \rvs $\bX$ is a DAG whose nodes are the \rvs\ and whose directed edges represent conditional dependencies.
It represents in a compact and factorized way a joint probability distribution $\Pr(\bX)$, by exploiting conditional (in)dependencies in the distribution.
This is achieved by associating with each node a \emph{function}---a conditional probability table---which quantifies the dependencies of each node from its parents.

Formally, a \textdef{Bayesian network} $\bn$ over the set of \rvs $\bX$ is a pair $(\G,\Phi)$ where:
\begin{itemize}
\item $\G$ is a directed acyclic graph (DAG), whose set of nodes is $\bX$;
\item $\Phi$ assigns to each $X\in\bX$ a \emph{conditional probability table (\cpt)} $\phi^X$, for $X$ given its parents.
\end{itemize}

\begin{theorem}[Semantics of BNs~\cite{Pearl86}]
\label{lem:BNsem}
A BN $\bn$ over the set of \rvs $\bX$ defines a unique probability distribution $\Pr(\bX)$ given by the \emph{product} of the CPTs associated to the nodes:
\begin{equation}
\label{eq:BNsem}
\tag{$\ast$}
\Pr(\bX) = \prod_{X\in \bX} \phi^X
\end{equation}
\end{theorem}

\begin{remark}[Markov condition]
Given a Bayesian network over $\bX$, each variable $X$ is assumed to be independent of its non-descendants, given its parents.
A probability distribution $\Pr(\bX)$ can be expressed as (\ref{eq:BNsem}) if and only if it satisfies such a Markov condition.
\end{remark}

The \emph{marginal distribution} over the \emph{variables of interest} $\bY \subseteq \bX$ is obtained by summing-out the irrelevant variables:
$\Pr(\bY)= \sum_{\bX\difset\bY} \Pr(\bX)$.
 
\subparagraph{From CPTs to Factors.}

The strength of BNs is to represent large probability distributions compactly, and to allow for inference algorithms that do not construct the full joint distribution explicitly.
CPTs alone are insufficient for this purpose, motivating the definition of a class of functions called \emph{factors}, a general mathematical abstraction representing both the initial parameters (the CPTs) and the partial results generated during the inference.
Formally, a \emph{factor} $\someqfact$ over the \rvs\ $\bX$ is a function mapping each $\bx\in\setofrv{\bX}$ to a non-negative real:
\begin{equation}
\label{eq:factor}
\someqfact:\setofrv{\bX}\to\nnReal
\end{equation}
A CPT $\theta$ for $X$ given $\bY$ is a factor over $X,\bY$ which satisfies $\sum_x \theta({x,\by}) = 1$ for all $\by\in\setofrv{\bY}$.

Unlike CPTs, factors are closed under both \emph{product and sum-out} operations.
This is crucial for algorithms (\eg\ Variable Elimination) that rely on pushing the sum (marginalization) to smaller components to reduce the cost of computations.
The so generated partial results are factors, but not CPTs.
Moving to factors is also necessary to incorporate evidence.

\subparagraph{Sharing variables.}

We refer to~\cite{DarwicheHandbook} for the formal definitions of sum and product of factors, which are the natural extensions of the corresponding operations for probabilities.
Tractability relies on the properties of these operations, and specifically on two key aspects: the product of factors \emph{inherently shares variables}, and product and sum \emph{distribute} under suitable conditions. Distributivity allows for pushing the sum on smaller components, reducing the size of computations.
Let us briefly discuss how sharing is built in the definition of product, and what this means. 
Assume we want to compute the distribution $\Pr(X,Y,Z)$ defined as the product $ \Pr(Z) \Pr(X \mid Z) \Pr(Y \mid Z)$.
This amounts to computing $2^3$ entries:
\begin{equation*}
\Pr(x,y,z) = \Pr(z) \Pr(x \mid z) \Pr(y \mid z) \quad \text{ for } x\in \setofrv{X}, y\in \setofrv{Y}, z\in \setofrv{Z}
\end{equation*}
Observe how the values corresponding to the same \rv\ $Z$ are shared: we compute $\Pr(z^\true) \Pr(x \mid z^\true) \Pr(y \mid z^\true)$ and $\Pr(z^\false) \Pr(x \mid z^\false) \Pr(y \mid z^\false)$, but never $\Pr(z^\true) \Pr(x \mid z^\true) \Pr(y \mid z^\false)$!

It is important to stress that the definition of the factors product is starkly different with that of the tensor product of matrices, which is the natural product for quantum systems.
We compare the two operations in \cref{app:section2}---the \emph{computational cost} being also starkly different (as discussed in~\cite[Example~9.7]{popl24}).

\subsection{Quantum data}

An (isolated) quantum system is associated with a complex vector space equipped with an inner product $\braket{\cdot | \cdot}$, \ie\ a \emph{Hilbert space} $\HH$ called the \definitive{state space} of the system.
The simplest quantum system is the \emph{qubit}, which has a two-dimensional state space corresponding to $\C^2$.

The \textbf{states} of a quantum system can be described in two mathematically equivalent ways, either as \emph{vectors} \emph{in} the space $\HH$ or as \emph{linear operators} acting \emph{on} the space $\HH$.
The latter formulation---standard in physics---shines in the description of quantum systems whose state is not known, as they allow to present a state as a probability distribution of states.

\begin{example}
\label{ex:preparation}
Intuitively, a state corresponds to a \emph{preparation}, \ie\ an action performed to set up an experiment.
Given two states $\rho_0$ and $\rho_1$, there must be a state $r\rho_0 + (1-r)\rho_1$ for any $r\in[0,1]$.
Indeed, a valid way to build a state from two preparations $\rho_0$ and $\rho_1$ is the preparation ``flip a fair coin: if heads, continue as the preparation $\rho_0$, otherwise continue as the preparation $\rho_1$'', yielding $\rho = \tfrac12\rho_0 + \tfrac12\rho_1$.
A probability distribution of states (such as $\rho$) also arises naturally to express \emph{partial knowledge} (or \emph{degrees of belief}) about the state of a system.
For example, $\rho$ above expresses we do not know if the state of the system is $\rho_0$ or $\rho_1$.
\end{example}

Before giving further details, we need to recall some notions from linear algebra.

\subparagraph{Hilbert Spaces.}

A finite $n$-dimensional Hilbert space $\HH$ is a complex vector space (isomorphic to) $\C^n$ with an inner product $\braket{\cdot |\cdot}$.
Column vectors are written $\ket{\psi}$, with $\psi$ a label, and their conjugate-transpose $\bra{\psi} \defeq \ket{\psi}^\dagger$. 	
\Cref{fig:algebra} (from~\cite{Nielsen_book}) summarizes relevant notations.

\begin{table}
\begin{tabular}{|>{\centering$}m{3cm}<{$}|m{10cm}|}
\hline
\textbf{Notation} & \textbf{Description}
\\
\hline
z^* & Complex conjugate of the complex number $z$.
\\
\hline
\ket{\psi} & Vector. Also known as a \textit{ket}.
\\
\hline
\bra{\psi} & Vector dual to $\ket{\psi}$. Also known as a \textit{bra}.
\\
\hline
\braket{\varphi | \psi} & Inner product between the vectors $\ket{\varphi}$ and $\ket{\psi}$.
\\
\hline
\ket{\varphi} \otimes \ket{\psi} & Tensor product of $\ket{\varphi}$ and $\ket{\psi}$. 
\\
\hline
\ket{\varphi \psi} & Abbreviated notation for $\ket{\varphi} \otimes \ket{\psi}$.
\\
\hline
\M^T & Transpose of the matrix $\M$.
\\
\hline
\M^{\dagger} & Adjoint (\aka\ Hermitian conjugate, conjugate-transpose) of $\M$:

$
\left[\begin{smallmatrix}
a\phantom{^*} & b \\ c\phantom{^*} & d
\end{smallmatrix}\right]^{\dagger}
=
\left[\begin{smallmatrix}
a^* & c^* \\ b^* & d^*
\end{smallmatrix}\right]
$.
\vspace*{.2em}
\\
\hline
\braket{\varphi | \M | \psi} & Inner product between $\ket{\varphi}$ and $\M \ket{\psi}$.

Equivalently, inner product between $\M^{\dagger} \ket{\varphi}$ and $\ket{\psi}$.
\\
\hline
\end{tabular}
\caption{Main linear algebra notations.}
\label{fig:algebra}
\end{table}

Given a Hilbert space $\HH$, $\L(\HH)$ denotes the set of \definitive{linear operators} \emph{on} $\HH$ (\ie\ from $\HH$ to $\HH$), which we 
identify with their matrix representations.
The \definitive{adjoint} $\M^\dagger$ of a matrix $\M$ is its conjugate-transpose.
We denote by $\id[\L(\HH)]$ the identity matrix in $\L(\HH)$.
The \definitive{trace} of a square matrix $\M = (\M_{ij})_{ij}\in \C^{n\times n}$ is $\tr{\M} = \sum_i \M_{ii}$.

A relevant subclass of $\L(\HH)$ is the set of positive operators $\Pos{\HH}$, playing the role of ``non‑negative real numbers''.
A linear operator $\M$ is \textbf{positive} 
if $\braket{\psi | \M | \psi} \in \nnReal$ for any $\ket \psi \in \HH$.

\begin{remark}
\label{rem:pos}
In the 1-dimensional space $\C^1$, the set of positive matrices in $\L(\C)=\C^{1\times 1}=\C$ is $\Pos{\C}=\nnReal$.
\end{remark}

\subparagraph{Quantum states as vectors in $\HH$.}

A common way to represent the states of a quantum system is by unit vectors in $\HH$, \ie\ vectors $\ket{\psi}$ such that $\braket{\psi | \psi} = 1$.
Considering a qubit, the vectors $\ket{0} = \left[\begin{smallmatrix}1\\0\end{smallmatrix}\right]$ and $\ket{1} = \left[\begin{smallmatrix}0\\1\end{smallmatrix}\right]$ form an orthonormal basis for its state space $\C^2$: an arbitrary state vector can be written $\ket{\psi} = \alpha \ket{0} + \beta \ket{1}$.
Intuitively, the states $\ket{0}$ and $\ket{1}$ of a qubit are analogous to the two values $0$ and $1$ which a \emph{bit} may take.
The way a qubit differs is that \emph{superpositions} (\ie\ linear combinations) of these two states also exist.
Some important states are $\ket{+} = \frac{1}{\sqrt{2}}(\ket{0} + \ket{1})$ and $\ket{-} = \frac{1}{\sqrt{2}}(\ket{0} - \ket{1})$, which form the Hadamard basis.

\subparagraph{Quantum states as density operators acting on $\HH$.}

Another usual way in physics to describe the state of a quantum system is by a \emph{positive operator $\rho$ with trace $1$}---called a \definitive{density operator}---that \emph{acts} on the state space $\HH$ of the system.
We denote by $\D(\HH) \subset \Pos{\HH}$ the set of density operators acting on $\HH$.
A vector $\ket{\psi}$ is represented by the matrix $\ket{\psi}\bra{\psi}$, called a \emph{pure state}.
If a quantum system is in the state $\rho_{i}$ with probability $p_i$, the density operator for the system is $\sum_{i} p_{i} \rho_{i}$.
This formalizes \Cref{ex:preparation}: if we can prepare qubits in pure states, we can also prepare qubits in any probability distribution of states.

\subparagraph{Composite systems.}

The \emph{state space} of a composite physical system is the tensor product of the state spaces of its components.
Given several systems, with the states $(\rho_i)_{i\in I}$, the \emph{state} of the composite system is $\bigotimes_{i\in I}\rho_i$.
Consider the Hilbert spaces $\HH_1$ with $\{\ket{\psi_i}\}_{i\in I}$ one of its bases, and $\HH_2$ with $\{\ket{\phi_j}\}_{j\in J}$ one of its bases.
Their tensor product $\HH_1 \otimes \HH_2$ is the Hilbert space with basis $\{\ket{\psi_i} \otimes \ket{\phi_j}\}_{(i,j)\in I\times J}$.
A short notation for $\ket{\psi}\otimes\ket{\phi}$ is $\ket{\psi\phi}$.
A pure quantum state in $\HH_1 \otimes \HH_2$ is \emph{separable} when it can be expressed as the tensor product of two vectors of $\HH_1$ and $\HH_2$.
Otherwise, it is \emph{entangled}, as for example the Bell state $\frac{1}{\sqrt{2}}(\ket{00} + \ket{11})$.

The \definitive{partial trace} $\ptr{\HH_1}$ of an operator in $\L(\HH_1\otimes\HH_2)$ is an operator in $\L(\HH_2)$ defined as $\ptr{\HH_1}\left(\M\otimes\M'\right)=\tr{\M}\M'$ for $\M\in\L(\HH_1)$ and $\M'\in\L(\HH_2)$, and extended to the general case by linearity.
The partial trace allows to ``\emph{forget}'' or ``\emph{discard}'' part of a composite system.

\subparagraph{Quantum operations (\aka\ quantum channels).}

States $\M\in\D(\HH)$ transform as $\M' = \E(\M)$, where the map $\E: {\D}(\mathcal{H}) \to {\D}(\mathcal{H'})$, called \definitive{quantum operation} or \definitive{quantum channel}, is:
\begin{enumerate}
\item\label{item:qo_cp}
\emph{Completely Positive}:
$\E$ sends positive operators to positive operators, and do so also when acting on a sub-part of a larger, potentially \emph{entangled}, composite system; \ie\ any ``extension'' $\idop\otimes \Eop$ (with $\idop$ the identity) \emph{carries a positive operator to a positive one}.
\item\label{item:qo_tp}
\emph{Trace-Preserving}:
if $\tr{\rho}=1$ then $\tr{\E(\rho)}=1$.
\item\label{item:qo_l}
\emph{Convex-Linear}:
$\E$ acts on a mixture of states by acting on each component individually, \ie\ $\Eop\left(\sum_{i\in I}p_{i}\M_{i}\right) = \sum_{i\in I}p_{i}\Eop(\rho_{i})$.
\end{enumerate}
By \ref{item:qo_cp} and \ref{item:qo_tp}, $\Eop$ (or more generally $\idop\otimes\Eop$) sends a density operator to a density operator.

\begin{example}
Two fundamental ways to modify a quantum state are \emph{unitary transformations}, describing the evolution of an isolated system, and \emph{measurements}.
A matrix $U\in\L(\HH)$ is unitary when $UU^\dagger = U^\dagger U = \id$.
The corresponding quantum channel is $\E:\rho \mapsto U \rho U^\dagger$.
Similarly, measurement operators $M_m$ act on a state $\rho$ as $M_m \rho M_m^\dagger$.
One more example of quantum operation is the partial trace.
\end{example}

\subparagraph{Quantum Instruments.}

Quantum channels describe the evolution of quantum states, \emph{ignoring classical information}---such as measurement outcomes---which plays a central role in ``real'' (or thought) experiments, as in \Cref{fig:Bell}.
A generalization of channels are \emph{quantum instruments}, mathematical devices that model ``real'' devices by specifying a post-measurement state per classical outcome.
This provides a framework where measurements are operations producing both quantum state updates and classical data.
A \textbf{quantum instrument} is an indexed collection $\{\E_a\}_a$ of completely positive linear maps that sum to a trace preserving map (\ie\ to a quantum channel).
Intuitively, the labels $a$ are the \emph{classical} outcomes ($a^\true/a^\false$) of a measurement---think $a\in \setofrv{A}$.
An instrument $\{\mathcal{E}_{a}\}_{a}$ induces a quantum channel $\mathcal{E}(\rho) =\sum_{a} \mathcal{E}_{a}(\rho)$, with $\Pr(a \mid \rho) = \tr{\mathcal{\E}_{a}(\rho)}$ providing the outcome probabilities.
We refer \eg\ to~\cite[page~112]{Watrous_book} for details.

\begin{remark}[Pure quantum and pure classical]
\label{rem:qi_pr}
The two limit cases yield standard notions.
\begin{itemize}
\item
Quantum instruments with no classical outcome (\ie\ with a single label $a$) are simply quantum operations.
\item
Quantum instruments from $\C$ to $\C$ are (isomorphic to) probability distributions.
\end{itemize}
\end{remark}

\subsection{Quantum Bayesian Networks (Instrument-based)}
\label{sec:qbn_instrument}

The generalization of Bayesian networks in~\cite{GBN2014} is based on the notion of quantum instruments. In \Cref{sec:qbn}, we will provide a definition of quantum Bayesian Network in terms of \qfactor s---the original definition (to which ours is equivalent) is not necessary to the reader. 

In this background section, we briefly review the key ideas in~\cite{GBN2014}, and what are the issues with the semantics there. 
The causal structure is described by a DAG (as in \cref{fig:Bell}); to include quantum sources of causality, the set of nodes is extended with new unobserved nodes, corresponding to quantum systems.
To each node of the generalized DAG is associated a \emph{family of quantum instruments} rather than a CPT, as we illustrate in \cref{ex:Alice_instrument}.

\begin{remark}
Recall that to each BN node is associated a CPT, \ie\ a \emph{family of probability distributions} indexed by values.
Intuitively, quantum instruments generalize probability distributions, so a \emph{family of quantum instruments} indexed by values generalizes a CPT.
\end{remark}

\begin{example}[Alice's instruments]
\label{ex:Alice_instrument}
Referring to \Cref{fig:Bell}, we want to model that Alice (the node $A$) performs the following operations, depending on the outcome of $X$.
Given $x^t$, Alice measures the qubit $q_1$ on the $(\ket{0},\ket{1})$ basis, returning $a^t$ if she obtains $\ket{0}$ and $a^f$ if not.
Given $x^f$, Alice measures $q_1$ on the $(\ket{-},\ket{+})$ basis, returning $a^t$ if she obtains $\ket{+}$, $a^f$ if not.

To model this setting, we associate to $A$ a \emph{family} of instruments, indexed by $\Val{X}$: $\{\{\Eop_{a^t,x^t},\Eop_{a^f,x^t}\},\{\Eop_{a^f,x^f},\Eop_{a^f,x^f}\}\}$.
The instrument given $x^t$ is $\{\Eop_{a^t,x^t},\Eop_{a^f,x^t}\}$ (from $\L(\HH_{q_1})$ to $\C$) with $\Eop_{a^t,x^t} : \M \mapsto \ket{0}\bra{0}\M\ket{0}\bra{0}$ and $\Eop_{a^f,x^t}: \M \mapsto \ket{1}\bra{1}\M\ket{1}\bra{1}$.
The instrument given $x^f$ is $\{\Eop_{a^f,x^f},\Eop_{a^f,x^f}\}$, 
with $\Eop_{a^t,x^f} : \M \mapsto \ket{+}\bra{+}\M\ket{+}\bra{+}$ and $\Eop_{a^f,x^f}: \M \mapsto \ket{-}\bra{-}\M\ket{-}\bra{-}$.
\end{example}


\subparagraph{Semantics of \QBN s.} 

The interpretation of a \QBN\ $\somebayes$ over a set of classical variables $\bX$ and a set of quantum systems $\bQ$ is defined in~\cite{GBN2014} as a suitable product of the instruments associated to the nodes.
Similarly to BNs, the semantics $\sem{\somebayes}$ of $\somebayes$ is a probability distribution over $\bX$, reflecting the fact that quantum states cannot be observed directly---we can only observe the classical outcomes of measurements.
This is exactly what happens in the Bell set-up (\Cref{ex:bell}), where we compute $\Pr(a,b\mid x,y)$. 

\subparagraph{Lack of compositionality.}
\label{sec:lack}

While intuitive and natural, the approach via quantum instruments has a limit in \emph{the lack of compositionality}.
Indeed, the interpretation of a \QBN\ $\somebayes$ in~\cite{GBN2014} needs to follow a specific order (see ~\cite[page~12]{GBN2014}).
The definition of an instrument does not adapt well to interpreting arbitrary sub-components (as in \cref{fig:compositional}).
Given a partition of $\somebayes$ into $\somebayes_1$ and $\somebayes_2$, even if we can define $\sem{\somebayes_1}$ and $\sem{\somebayes_2}$, we do not have a natural composing function yielding $\sem{\somebayes_1} \bullet \sem{\somebayes_2}=\sem{\somebayes}$.
There are two distinct issues.
\begin{enumerate}
\item
Consider \Cref{fig:compositional} with its \textcolor{red}{left} and \textcolor{blue}{right} sub-graphs.
The \textcolor{red}{left} sub-graph should be interpreted by an instrument $\{\Eop^l_x \mid x \in\setofrv{X}\}$ with $\Eop^l_x : \HH_2\to \HH_1$, and the \textcolor{blue}{right} sub-graph by $\Eop^r : \HH_1\to \HH_2$.
The interpretation of the full graph is an instrument $\{\Eop_x \mid x \in\setofrv{X}\}$ with $\Eop_x:\C\to\C$. 
However, the operations on instruments do not immediately yield from $\Eop^l_x$ and $\Eop^r$ a function of type $\C\to\C$.
\item
A further issue is the impossibility (in general) of dealing by components with hidden variables, exactly as happens for CPTs, which are closed by product but not by sum.
\end{enumerate}

\begin{remark}[Marginals and hidden variables]
\label{rem:marginals}
A technical but crucial point when considering compositionality is the variables of interest---or dually, the hidden variables.
In \cref{fig:compositional}, the shaded node $Y$ corresponds to a hidden (or unobserved) variable.
What we want here is the marginal $\Pr(x)$, which can be obtained from $\Pr(x,y)$ by summing-out the irrelevant variable: $\sum_{y \in \setofrv{Y}} \Pr(x,y)$.
When computing the semantics of the \textcolor{blue}{right} component, it would be desirable to already sum-out $Y$.
This kind of operation is regularly performed in Bayesian inference algorithms to reduce the size of computations.
\end{remark}

Addressing both issues with compositionality motivates the framework of \emph{quantum factors} that we introduce in the next section.
In the classical setting, abstracting CPTs into \emph{factors} gives a status to partial calculations. We will follow a similar path.

%% file: 02_QuantumDistributions.tex
\section{Quantum Factors}
\label{sec:semantic}

A \QBN\ is a DAG with nodes \emph{a set $\bX$ of \rvs} and a \emph{set $\bQ$ of quantum systems}.
We associate to each node a \emph{function}---called quantum factor---from $\Val{\bX}$ to \emph{positive operators} on the relevant Hilbert space.
We then interpret in the same way \QBN s, as well as any sub-component.
In this section we define quantum factors (\Qfactor s).
In \Cref{sec:qbn}, we reformulate the definition of \QBN s in this setting, and explain the equivalence with~\cite{GBN2014}.
Let us illustrate the idea with some examples.

\begin{example}
Consider \Cref{ex:preparation}, the preparation of a single-qubit quantum system $\qreg$ depending on a classical \rv\ $X$ (\eg\ a fair coin).
This is described by the following DAG:
\begin{tikzpicture}[scale=0.7,baseline=(b)]
\begin{genscope}
	\coordinate (b) at (0,-.2);
	\node[minimum size=5mm] (X1) at (0,0) {$X$};
	\node[minimum size=5mm] (Q1) at (1.4,0) {$\qreg$};
	\path (X1) edge (Q1);
\end{genscope}
\end{tikzpicture}.
The interpretation of this system is a function $\phi: \Val{X} \to \Pos{\HH_Q}$ giving for each $x \in \Val{X}$ a positive matrix $\M_{x}$ on $\HH_\qreg$.
If we are interested only in the state of the qubit, we can sum-out $X$ to obtain the density matrix $\Pr(x^t) \M_{x^t}+ \Pr(x^f) \M_{x^f}$, as in \Cref{ex:preparation}.

As another example, the preparation of a system $\qreg$ depending on two classical \rvs\ $X$ and $Y$ is a function from $\Val{X,Y}$ to $\Pos{\HH_\qreg}$, yielding four positive matrices $\M_{xy}\in\Pos{\HH_\qreg}$.
\end{example}

\begin{example}[Alice \& Bob]
Consider the DAG in \Cref{fig:Bell}, with four classical variables $A$, $B$, $X$ and $Y$, and one quantum node $\qreg$ producing two entangled qubits $q_1$ and $q_2$ in the state space $\HH_{q_1}\tens\HH_{q_2}$.
We associate to each node a function as follows.
\begin{itemize}
\item
To $X$ a function $\phi^X: \setofrv{X}\to\Pos{\HH_{\emptyset}}=\nnReal$ (recall \Cref{rem:pos}).
Similarly for $Y$.
\item
To $\qreg$ a positive matrix in $\Pos{\HH_{q_1}\tens\HH_{q_2}}$.
\item
To $A$ a function $\phi^A:\setofrv{X}\times\setofrv{A} \to \Pos{\HH_{q_1}}$.
\end{itemize}
\end{example}


\subparagraph*{Quantum registers.}

For simplicity's sake, from now on we assume quantum systems to be composed of qubits.
Assuming a countable set $\mathbb{Q}$ of qubits $q_0,q_1,\dots$ with associated 2-dimensional Hilbert spaces $\HH_{q_0},\HH_{q_1},\dots$, a \textdef{quantum register} is a \emph{finite} set of qubits $\qreg \subseteq \mathbb{Q}$.
Given a quantum register $\qreg$ of $n$ qubits, its associated Hilbert space is $\HH_{\qreg} = \bigotimes_{q \in \qreg} \HH_{q} \cong \C^{2^n}$.
We write $\qreg \uplus \qreg'$ the \emph{disjoint union} of the registers $\qreg$ and $\qreg'$, only defined when $\qreg \cap \qreg' = \emptyset$.

\subparagraph*{Quantum Factors.}

Please compare (\ref{eq:qfactor}) below with (\ref{eq:factor}) the definition of a factor in \Cref{sec:BN}.

\begin{definition}[Quantum Factor (\qfactor)]
Let $\bX$ be a set of random variables and $\qreg$ a quantum register.
A \definitive{quantum factor} (shortened as \definitive{\qfactor}) $\phi$ over $(\bX,\qreg)$ is a function
\begin{equation}\label{eq:qfactor}
\someqfact: \setofrv{\bX} \to \Pos{\HH_\qreg}
\end{equation}
mapping each tuple of values $\bx$ to a \emph{positive} operator in $\L(\HH_\qreg)$.
The \definitive{scope} of $\someqfact$ is $\bX\cup\qreg$.
\end{definition}

\begin{remark}[Pure classical and pure quantum]
The two limit cases yield familiar notions.
\begin{itemize}
\item
If $\qreg$ is empty, then positive operators in $\L(\C^1)$ are exactly elements of $\nnReal$, yielding the standard definition of a factor over a set of random variables.
\item
If $\bX$ is empty, then a \qfactor\ simply corresponds to \emph{a positive operator} in $\Pos{\HH_\qreg}$.
\end{itemize}
\end{remark}

The trivial \qfactor\ over $(\emptyset,\emptyset)$ simply corresponds to a \emph{non-negative} real.
We denote by $\ftone$ the trivial \qfactor\ corresponding to $1$.

\subsection{Product and Sum of Quantum Factors}

We show that \qfactor s admit and are \emph{closed} under \emph{product} $\odot$ and \emph{summing-out} $\sum$ operations.
Remarkably, \qfactor s encompasses both tensor networks and factors from Bayesian networks:
\begin{itemize}
\item
For \qfactor s $\someqfact$ and $\someqfact'$ over classical variables only, $\someqfact \compfact \someqfact'$ is the product of factors and $\sum_X \someqfact$ is summing-out (\aka\ \emph{marginalization}).
\item
For \qfactor s $\someqfact$ and $\someqfact'$ over quantum registers only, $\someqfact \compfact \someqfact'$ is the product of tensor networks 
and $\sum_\qreg \someqfact$ is tracing-out (\ie\ taking the \emph{partial trace} over $\HH_\qreg$).
\end{itemize}

\subparagraph*{Sum-out for \qfactor s.}
\emph{Summing out a variable} from a \qfactor\ corresponds to \emph{marginalization} for classical variables, and to the \emph{partial trace} for quantum systems.
\begin{definition}[Sum-out]
\label{def:sum}
Let $\someqfact$ be a \qfactor\ over $(\bX,\qreg)$. 
\begin{itemize}
\item
For $Z \in \bX$, the \definitive{sum $\sum_{Z} \someqfact $} is the \qfactor\ over $(\bX \difset\{Z\},\qreg )$ defined by:
\begin{equation*}
\left(\sum_Z \someqfact\right) (\by) \defeq \sum_{z\in\setofrv{Z}} \someqfact (\by,z)
\qquad
\text{for $\by \in \setofrv{\bX \difset\{Z\}}$}
\end{equation*}
\item
For $\qreg' \subseteq \qreg $, the \definitive{sum $\sum_{\qreg'} \someqfact $} is the \qfactor\ over $(\bX ,\qreg \difset\qreg')$ defined by:
\begin{equation*}
\left(\sum_{\qreg'} \someqfact\right) (\bx) \defeq \ptr{\HH_{\qreg'}}\left(\someqfact (\bx)\right)
\qquad
\text{for $\bx \in \setofrv{\bX}$}
\end{equation*}
\end{itemize}
\end{definition}

The definition of the product is more delicate, because classical variables and quantum systems behave differently: classical values can be ``shared'', while quantum states are linear resources.
To mathematically conciliate the two, we describe the product in terms of bases.

\subparagraph{Describing \qfactor s in terms of the canonical basis.}

For a qubit $q_i$, we write $\baseof{q_i} = (e^{00}_i,e^{01}_i,e^{10}_i,e^{11}_i)$ for the canonical orthonormal basis of $\L(\HH_{q_i}) \cong \C^{2\times 2}$, \ie\ $e^{00}_i = \ket{0}\bra{0}$, $e^{01}_i = \ket{0}\bra{1}$, $e^{10}_i = \ket{1}\bra{0}$, and $e^{11}_i = \ket{1}\bra{1}$.
Given a quantum register $\qreg = \{q_1, \dots, q_m\}$, we write $\baseof{\qreg} = \prod_{i=1}^m\baseof{q_i}$, which is an orthonormal basis for $\L(\HH_\qreg) = \L(\bigotimes_{i=1}^m \HH_{q_i})$.
A \qfactor\ $\someqfact$ over $(\bX,\qreg)$ is then associated to the following function, that we denote by $\qfactfunc{\someqfact}$:
\begin{align*}
\qfactfunc{\someqfact}: \setofrv{\bX} \times \baseof{\qreg} &\to \C\\
(x_1,\dots,x_n,e_1,\dots,e_m)&\mapsto \someqfact(x_1,\dots,x_n)_{e_1,\dots,e_m}
\end{align*}
where $\someqfact(x_1,\dots,x_n)_{e_1,\dots,e_m}$ is the scalar at position $(e_1,\dots,e_m)$ of the matrix $\someqfact(x_1,\dots,x_n)$.
Notice that $\qfactfunc{\someqfact}$ uniquely determines $\someqfact$, and vice-versa.

\begin{example}
Consider \rvs\ $X_1$ and $X_2$, and a quantum register $\{q_1,q_2\}$.
An example of \qfactor\ $\someqfact$ over $(\{X_1,X_2\},\{q_1,q_2\})$ is:
\begin{equation*}
\someqfact(x_1^t,x_2^t)
=
\someqfact(x_1^t,x_2^f)
=
\someqfact(x_1^f,x_2^t)
=
\left[\begin{smallmatrix}
	1 & 0 & 0 & 0 \\
	0 & 1 & 0 & 0 \\
	0 & 0 & 1 & 0 \\
	0 & 0 & 0 & 1 \\
\end{smallmatrix}\right]
\qquad
\someqfact(x_1^f,x_2^f)
=
\left[\begin{smallmatrix}
	2 & 0 & 0 & 1 \\
	0 & 2 & 0 & 0 \\
	0 & 0 & 2 & -1 \\
	1 & 0 & -1 & 2 \\
\end{smallmatrix}\right]
\end{equation*}
We also see $\someqfact$ as a function $\qfactfunc{\someqfact}:\setofrv{X_1}\times\setofrv{X_2}\times\baseof{q_1}\times\baseof{q_2} \to \C$ sending $(x_1,x_2,e_1,e_2)$ to the scalar corresponding to $e_1\tens e_2$.
For instance:
\begin{itemize}
\item
$\qfactfunc{\someqfact}(x_1^t,x_2^t,e^{00}_1,e^{00}_2) = \hphantom{-}1$,
the entry of $\someqfact(x_1^t,x_2^t)$ corresponding to $e^{00}_1 \tens e^{00}_2$;
\item
$\qfactfunc{\someqfact}(x_1^f,x_2^f,e^{11}_1,e^{01}_2) = -1$,
the entry of $\someqfact(x_1^f,x_2^f)$ corresponding to $e^{11}_1 \tens e^{01}_2$.
\end{itemize}
\end{example}

\subparagraph*{Product of \qfactor s.}
We now define the product.
It shares the values of classical variables (exactly as happens in the factors product, as discussed in \Cref{sec:BN}), while on quantum registers it behaves as the contraction of tensor networks. 
Below, we use $\qfactfunc{\someqfact}$ instead of $\someqfact$; the symmetric difference $S \symdif R \defeq (S \cup R) \difset (S \cap R)$ replaces the union in the quantum case.

\begin{definition}[Product of \qfactor s]
\label{def:prod}
Let $\someqfact_1$ and $\someqfact_2$ be \qfactor s respectively over $(\bX_1,\qreg_1)$ and $(\bX_2,\qreg_2)$.
The \definitive{product $\someqfact_1 \compfact \someqfact_2$} is a \qfactor\ over $(\bX_1\cup\bX_2,\qreg_1\symdif\qreg_2)$ defined as follows.
Let $\bx\in\setofrv{\bX_1\cup\bX_2}$ and $(\be_1,\be_2)\in\baseof{\qreg_1\symdif\qreg_2}=\baseof{\qreg_1\difset\qreg_2}\times\baseof{\qreg_2\difset\qreg_1}$.
Call $\bx_1$ the elements of $\bx$ belonging to $\setofrv{\bX_1}$, and similarly $\bx_2$ those in $\setofrv{\bX_2}$.
We set
\begin{equation*}
\qfactfunc{(\someqfact_1 \compfact \someqfact_2)} (\bx,\be_1,\be_2)
\defeq
\sum_{\be_3 \in \baseof{\qreg_1\cap\qreg_2}} \qfactfunc{\someqfact_1}(\bx_1,\be_1,\be_3)\ \qfactfunc{\someqfact_2}(\bx_2,\be_2,\be_3)
\end{equation*}
\end{definition}

\subparagraph*{Properties of sum and product.}
\qfactor s are \emph{closed under sum and product}.
Easy calculations show that the product is (essentially) associative, product and sum are both commutative, and---crucially---they distribute under suitable conditions, like classical factors.

\begin{restatable}{proposition}{qfactors}
\label{lem:qfactors_op}
\qfactor s are \emph{closed} under both product and sum.
Moreover:
\begin{enumerate}
\item\label{item:qfactors_op:1}
$(\someqfact_1\compfact \someqfact_2) \compfact \someqfact_3 = \someqfact_1 \compfact (\someqfact_2\compfact \someqfact_3)$ provided no qubit appears in the scope of all three \qfactor s
\item
$\someqfact_1 \compfact \someqfact_2 = \someqfact_2 \compfact \someqfact_1$
\item
$\sum_U \sum_V \someqfact_1 = \sum_V \sum_U \someqfact_1$
\item
$\sum_V(\someqfact_1 \compfact \someqfact_2) = (\sum_V \someqfact_1) \compfact \someqfact_2$ when $V$ is not in the scope of $\someqfact_2$
\end{enumerate}
\end{restatable}
\noindent
(Proof in \cref{app:section3}.)
The side condition in point \ref{item:qfactors_op:1} is always satisfied by quantum BNs.

%% file: 03_QBNs.tex
\section{Quantum Bayesian Networks (\Qfactor-based)}
\label{sec:qbn}

We define quantum Bayesian networks in terms of \qfactor s.
Our definition is equivalent to the instrument-based definition in~\cite{GBN2014}, however moving to \qfactor s allows the semantics to address compositionality of sub-networks.
To each node, we associate a special kind of \qfactor, called a \qCPT.
This ensures that the semantics of the full network is exactly a probability distribution, and not any arbitrary function with values in $\nnReal$.

\subparagraph{Quantum CPTs.}

We define a class of \Qfactor s which plays a similar role to that of CPTs.

\begin{definition}[\QCPT\ (\qCPT)]
\label{def:condition_qfactor_qbn}
A \qfactor\ $\someqfact$ over $(\bX,\qreg)$ is
\begin{itemize}
\item
a \definitive{\qCPT\ for $Z \in \bX$ given $(\bY,\qreg)$} if $\bX = \{Z\} \uplus \bY$, and $\forall \by\in\setofrv{\bY}$:
\begin{equation*}
\left(\sum_{Z} \someqfact\right)(\by)
=
\id[\L(\HH_\qreg)]
\end{equation*}
\item
a \definitive{\qCPT\ for $\qreg' \subseteq \qreg$ given $(\bX, \qreg'')$} if $\qreg = \qreg' \uplus \qreg''$, and $\forall \bx\in\setofrv{\bX}$:
\begin{equation*}
\left(\sum_{\qreg'} \someqfact\right)(\bx)
=
\id[\L(\HH_{\qreg''})]
\end{equation*}
\end{itemize}
\end{definition}

The two limit cases correspond to standard, familiar notions.
\begin{itemize}
\item
In the \emph{pure classical case} ($\qreg=\emptyset$), being a \qCPT\ amounts to being a CPT.
Indeed, a \qCPT\ for $X$ given $(\bY, \emptyset)$ is a \qfactor\ $\someqfact$ over $(\{X\}\uplus\bY,\emptyset)$ which satisfies $\sum_{x\in\setofrv{X}} \someqfact(x,\by) = 1$.
\item
In the \emph{pure quantum case} ($\bX=\emptyset$), being a \qCPT\ amounts to being \emph{Choi state}, a well-known condition on operators that we will detail in \cref{sec:choi_iso}.
\end{itemize}

\subparagraph{Quantum Bayesian Networks.}

We define quantum Bayesian networks as DAGs with a \qCPT for each node.

\begin{definition}[Quantum Bayesian Network (\qqbn)]
A \definitive{quantum Bayesian network} over a set $\bX$ of \rvs\ and a set $\bQ$ of disjoint non-empty quantum registers is a pair $(\G,\someqfactset)$ where:
\begin{itemize}
\item
$\G$ is a directed acyclic graph over the set of nodes $\bX\cup\bQ$;
\item
each out-edge of a node $\qreg\in\bQ$ is labelled by a non-empty register $\qreg_i$, with $\uplus_i \qreg_i = \qreg$;
\item
each out-edge of a node $X\in\bX$ is labelled by $X$;
\item
$\someqfactset$ assigns to each $V \in \bX\cup\bQ$ a \qCPT\ for $V$ given the \rvs\ and qubits labelling its in-edges.
\end{itemize}
\end{definition}

\begin{example}[Alice's \qCPT]
Consider \cref{fig:Bell} with the setting of \cref{ex:Alice_instrument}.
The node $A$ is given the \qfactor\ $\someqfact^A: \setofrv{A}\times\setofrv{X} \to \L(\HH_{q_1})$ below, that is a \qCPT\ for $A$ given $(\{X\},\{q_1\})$:

\begin{adjustbox}{}
\begin{tabular}{ccc@{ }c@{ }l@{}l@{\qquad\qquad}ccc@{ }c@{ }l@{}l}
$(a^t,x^t)$ & $\mapsto$ & $\ket{0}\bra{0}$ & $=$ &
&
$\left[\begin{smallmatrix}
1 & 0\\
0 & 0\\
\end{smallmatrix}\right]$
&
$(a^f,x^t)$ & $\mapsto$ & $\ket{1}\bra{1}$ & $=$ &
&
$\left[\begin{smallmatrix}
0 & 0\\
0 & 1\\
\end{smallmatrix}\right]$
\\
$(a^t,x^f)$ & $\mapsto$ & $\ket{+}\bra{+}$ & $=$ &
$\frac{1}{2}$
&
$\left[\begin{smallmatrix}
1 & 1\\
1 & 1\\
\end{smallmatrix}\right]$
&
$(a^f,x^f)$ & $\mapsto$ & $\ket{-}\bra{-}$ & $=$ &
$\frac{1}{2}$
&
$\left[\begin{smallmatrix}
1 & -1\\
-1 & 1\\
\end{smallmatrix}\right]$
\end{tabular}
\end{adjustbox}
\end{example}

\subsection{Semantics and Compositionality}

The semantics of a \qqbn\ is a probability distribution, defined akin to that of a BN (\cref{lem:BNsem}).

\begin{definition}
For $\mathcal{B} = (\G,\someqfactset)$ a \qqbn\ over \rvs\ $\bX$ and registers $\bQ$, its \definitive{semantics} is:
\begin{equation*}
\sem{\mathcal{B}} \defeq \bigodot_{V \in \bX\cup\bQ} \someqfactset(V)
\end{equation*}
\end{definition}

\begin{proposition}
\label{lem:sem_qbn_dist}
With the above notations, $\sem{\mathcal{B}}$ is a probability distribution over $\bX$.
\end{proposition}
The proof is obtained in a similar way to that for the classical result~\cite{Pearl86}.

The formalism allows us to naturally interpret sub-components of a \qqbn, as in \cref{fig:compositional}.
Call $(\G',\someqfactset')$ a \definitive{sub-network} of a \qqbn\ $(\G,\someqfactset)$ when $\G'$ is a sub-graph of $\G$ and $\someqfactset'$ is the restriction of $\someqfactset$ to nodes of $\G'$.
The semantics of a sub-network is defined as for the full network: the product of its \qfactor s.
Compositionality follows from associativity of $\compfact$.

\begin{proposition}[Compositionality]
\label{lem:compositionality_qbn}
Let $\mathcal{B} = (\G,\someqfactset)$ be a \qqbn, and $\mathcal{B}_1 = (\G_1,\someqfactset_1)$ and $\mathcal{B}_2 = (\G_2,\someqfactset_2)$ a partition in two sub-networks (\ie\ $\G = \G_1 \uplus \G_2$).
Then:
\begin{equation*}
\sem{\mathcal{B}} = \sem{\mathcal{B}_1} \compfact \sem{\mathcal{B}_2}
\end{equation*}
\end{proposition}


\subsection{Equivalence with Instrument-based Quantum Bayesian Networks}
\label{sec:choi_iso}

\subparagraph{Choi states and \qCPTs.}

The \definitive{Choi-Jamiolkowski isomorphism} (see \eg~\cite{Watrous_book}) is a classical, well-known result allowing to see linear maps from $\L(\HH_1)$ to $\L(\HH_2)$ as matrices in $\L(\HH_1 \otimes \HH_2)$.
It has many applications, basic ones being to check easily whether two quantum operations are equal---by checking their images are the same---or whether a map is completely positive; we refer to~\cite{Watrous_book} for details.
Let us denote by $\choiiso$ the Choi-Jamiolkowski isomorphism.
Given a quantum operation $\Eop : \L(\HH_1) \to \L(\HH_2)$, complete positivity of $\Eop$ is equivalent to positivity of $\choiiso(\Eop)$, and trace-preserving of $\Eop$ is equivalent to $\ptr{\HH_2}(\choiiso(\Eop)) = \id[\L(\HH_1)]$.
Hence, operators $\M\in\L(\HH_1 \otimes \HH_2)$ that are positive and such that $\ptr{\HH_2}(\M) = \id[\L(\HH_1)]$ are exactly the images by $\choiiso$ of quantum operations, and are called \definitive{Choi states}.
Please notice that being a \qCPT\ for a quantum register $Q'$, \ie\ the second item of \cref{def:condition_qfactor_qbn}, corresponds to being a Choi state (once values for the \rvs\ are fixed) with $\left(\sum_{\qreg'} \someqfact\right)(\bx) = \ptr{\HH_{\qreg'}}(\someqfact(\bx)) = \id[\L(\HH_{\qreg''})]$---the positivity requirement being already part of the definition of a \Qfactor.

\subparagraph{Equivalence with Instrument-based Quantum Bayesian Networks.}

Our definition of 
quantum Bayesian networks is equivalent to the instrument-based one of~\cite{GBN2014}.
Indeed, the definition of the DAGs is the same, and it is straightforward to provide a bijection between a \qCPT associated to a node in our formulation, and instruments associated to the same node in~\cite{GBN2014}.
The proof is similar to that for the Choi-Jamiolkowski isomorphism (see \cref{app:equivalence_qbns}).
The composition of instruments in~\cite{GBN2014} corresponds to the product of \qCPTs.

%% file: 04_QPNs.tex
\section{Quantum (Bayesian) Proof-Nets}
\label{sec:qpns}

In this section, we introduce \qpns.
They can be thought of as a typed version of \QBN s, built on the graph formalism of Multiplicative Linear Logic.
Any \QBN\ can be encoded as a \qpn\ and, conversely, to any \qpn of appropriate type is associated a \QBN.
What we gain by moving to the typed setting of proof-nets is:
\begin{itemize}
\item 
This formalism encodes \qqbn s and enables formalizing (and working with) sub-networks.
\item
Typing allows us to \emph{modularly compose} \qpns of compatible interfaces.
\item
Finally, we have \emph{a linear form of high-order}, as already sketched in \cref{fig:modular} where the network $\N_1$ of type $(A\lollipop C) \lollipop D$ is expecting as input a network of type $(A\lollipop C)$.
\end{itemize}

\subparagraph{Formulas.}

We assume given two countable sets of \textdef{names} denoted by metavariables $X, Y, \dots$ and $q_1, q_2, \dots$ which we call respectively \textdef{classical names} and \textdef{q-names}.
Formulas are those of the Multiplicative fragment of Linear Logic (\MLL):
\begin{equation*}
	\FormA,\FormB ::=
	X^+	\mid X^- \mid q^+ \mid q^-
	\mid\FormA\otimes\FormB
	\mid\FormA\parr\FormB
\end{equation*}

\textdef{Negation} $(\cdot)\orth$ is defined inductively by
$(X^+)\orth\defeq X^-$,
$(X^-)\orth\defeq X^+$, 
$(q^+)\orth \defeq q^-$,
$(q^-)\orth \defeq {q^+}$,
$(\FormA\otimes\FormB)\orth \defeq\FormA\orth\parr\FormB\orth$ and
$(\FormA\parr\FormB)\orth \defeq\FormA\orth\otimes\FormB\orth$.
\textdef{Linear arrow} is defined as usual by $F\lollipop G \defeq F^\bot \parr G$. 
We reserve the notation $\qreg^+$ for formulas of the shape $\otimes_{i\in I} q_i^+$.
Capital Greek letters $\Gamma,\Delta,\dots$ vary over finite sequences of formulas, \ie\ $\Delta = \FormA_1,\dots, \FormA_n$.
By $\Nm F$ we denote the set of all names which appear in the formula $F$, and similarly for $\Nm \Delta$.
For example, $\Nm{X^-\parr (Y^+\otimes X^+), q^+\parr q^-} =\{X,Y,q\}$.

A key property of formulas in Linear Logic is (positive/negative) \textdef{polarity}, which we indicate for atoms and extend to formulas as follows, defining (strictly) polarized formulas:
\begin{align*}
\textbf{Positive formulas: }& P, P'::= X^+ \mid q^+ \mid P \otimes P'
\\
\textbf{Negative formulas: }& N, N'::= X^- \mid q^- \mid N \parr N'
\end{align*}

\begin{remark}\label{rem:polarity}
Dual polarities carry a connotation of output/input, answer/question, active/ passive, player/opponent, etc.
Negative atoms are seen as \emph{inputs}, and positive ones as \emph{outputs}.
An intuition---at the base of the Linear Logic interpretation of the information flow---is that information travels upwards (\resp\ downwards) on negative (\resp\ positive) atoms.
\end{remark}

\subparagraph{MLL proof-nets (with boxes).}

In Linear Logic, proofs admit both a \emph{sequent calculus} syntax in the form of \textdef{\pts} and a \emph{graph syntax}, in the form of \textdef{\pn}.
Here, we are interested in the latter.
We omit the sequent calculus presentation---which is straightforward to deduce for the reader familiar with Linear Logic, but not relevant for our purpose.

\begin{definition}[Proof-Nets]
\label{def:pn}
We call \textbf{typed} a labelled (partial\footnote{A partial graph may have \emph{pending edges}, which may be either out-edges (as in \cref{ex:bell_pn}) or in-edges.}) graph $\somepn$ built from the alphabet of nodes on \cref{fig:grammar_nodes_pn}, with edges labelled by \MLL formulas.
Edges incident to a node $\somenode$ need to respect the typing conditions of the grammar, and are classified either as \definitive{premises} of $\somenode$ (depicted above $\somenode$) or as its \definitive{conclusions} (depicted below $\somenode$).
We further require that an edge is the conclusion of at most one node and the premise of at most one node.
\begin{itemize}
\item
\textbf{Proof-nets.}
A typed graph $\somepn$ is a \definitive{proof-net} if it has no pending premises (\ie\ every edge in $\somepn$ is the conclusion of some node) and if it is \definitive{correct}:
every cycle in $\somepn$ uses the two premises of a same $\parr$-node or the two premises of a same $\cc$-node.
\item
\textbf{Conclusions.}
The labels of the edges which are premise of no node (\aka\ pending conclusions) are called the \definitive{conclusions} of $\somepn$.
The \definitive{type} of a \pn $\somepn$ of conclusions $\Delta=C_1, \dots, C_k$ is the sequent $\vdash \Delta$.
We also write $\somepn \der \Delta$.
\end{itemize}
\end{definition}

\begin{example}
The typed graph in \cref{ex:bell_pn}---where $\qreg^+=q_1^+\tens q_2^+$---is a \pn.
\end{example}

The grammar in \cref{fig:pn} extends the standard nodes of \MLL\ with a new node, called a \textbf{box}, which has \emph{exactly one positive conclusion}---its \definitive{\main}---of type either $\X$ or $Q^+=\otimes_{i\in I} q_i^+$.
The negative conclusions---called \definitive{inputs}---are atomic.
All atoms appearing in the conclusions of a box are pairwise distinct.
The grammar is similar to that introduced in~\cite{EhrhardFP23}.
However, while in~\cite{EhrhardFP23} to each box is associated a CPT, here we associate a \qCPT.

\begin{figure}
\begin{adjustbox}{}
\begin{tikzpicture}
\begin{scope}[every node/.style={draw=none}, every path/.style={draw=black}]
	\node (ax) at (0,0) {$\ax$};
	\coordinate (axc1) at (-.75,-.5);
	\coordinate (axc2) at (.75,-.5);
	\path (axc1) |- node[named edge,pos=.25]{$X^-$} (ax) -| node[named edge,pos=.75]{$X^+$} (axc2);
\end{scope}
\end{tikzpicture}
\enskip
\begin{tikzpicture}
\begin{scope}[every node/.style={draw=none}, every path/.style={draw=black}]
	\node (ax) at (0,0) {$\ax$};
	\coordinate (axc1) at (-.75,-.5);
	\coordinate (axc2) at (.75,-.5);
	\path (axc1) |- node[named edge,pos=.25]{$q^-$} (ax) -| node[named edge,pos=.75]{$q^+$} (axc2);
\end{scope}
\end{tikzpicture}
\enskip
\begin{tikzpicture}
\begin{scope}[every node/.style={draw=none}, every path/.style={draw=black}]
	\node (ax) at (0,0) {$\cut$};
	\coordinate (axc1) at (-.75,.5);
	\coordinate (axc2) at (.75,.5);
	\path (axc1) |- node[named edge,pos=.25,left]{$A\orth$} (ax) -| node[named edge,pos=.75,right]{$A$} (axc2);
\end{scope}
\end{tikzpicture}
\enskip
\begin{tikzpicture}
\begin{scope}[every node/.style={draw=none}, every path/.style={draw=black}]
	\node (c) at (0,0) {$\tens$};
	\coordinate (cp1) at (-.5,.5);
	\coordinate (cp2) at (.5,.5);
	\coordinate (cc1) at (0,-.75);
	\path[out=-90,in=180] (cp1) edge node[named edge,left]{$A$} (c);
	\path[out=-90,in=0] (cp2) edge node[named edge,right]{$B$} (c);
	\path (c) -- node[named edge]{$A\tens B$} (cc1);
\end{scope}
\end{tikzpicture}
\enskip
\begin{tikzpicture}
\begin{scope}[every node/.style={draw=none}, every path/.style={draw=black}]
	\node (c) at (0,0) {$\parr$};
	\coordinate (cp1) at (-.5,.5);
	\coordinate (cp2) at (.5,.5);
	\coordinate (cc1) at (0,-.75);
	\path[out=-90,in=180] (cp1) edge node[named edge,left]{$A$} (c);
	\path[out=-90,in=0] (cp2) edge node[named edge,right]{$B$} (c);
	\path (c) -- node[named edge]{$A\parr B$} (cc1);
\end{scope}
\end{tikzpicture}
\enskip
\begin{tikzpicture}
\begin{scope}[every node/.style={draw=none}, every path/.style={draw=black}]
	\node (c) at (0,0) {$\cc$};
	\coordinate (cp1) at (-.5,.5);
	\coordinate (cp2) at (.5,.5);
	\coordinate (cc1) at (0,-.75);
	\path[out=-90,in=180] (cp1) edge node[named edge,left]{$\negatom{X}$} (c);
	\path[out=-90,in=0] (cp2) edge node[named edge,right]{$\negatom{X}$} (c);
	\path (c) -- node[named edge]{$\negatom{X}$} (cc1);
\end{scope}
\end{tikzpicture}
\enskip
\begin{tikzpicture}
\begin{scope}[every node/.style={draw=none}, every path/.style={draw=black}]
	\node (w) at (0,0) {$\cw$};
	\coordinate (cc1) at (0,-.75);
	\path (w) -- node[named edge]{$X^-$} (cc1);
\end{scope}
\end{tikzpicture}
\enskip
\begin{tikzpicture}
\begin{scope}[every node/.style={draw=none}, every path/.style={draw=black}]
	\node (w) at (0,0) {$\cw$};
	\coordinate (cc1) at (0,-.75);
	\path (w) -- node[named edge]{$q^-$} (cc1);
\end{scope}
\end{tikzpicture}
\enskip
\begin{tikzpicture}
\begin{scope}[every node/.style={rectangle,draw=black}, every path/.style={draw=black}]
	\node[minimum width=25mm,minimum height=6mm] (b) at (0,0) {$X/\qreg$};
	\coordinate (bc1) at (-1,-1);
	\coordinate (bc2) at (-.2,-1);
	\coordinate (bcn) at (1,-1);
	\node[draw=none] at (-.6,-.5) {$\dots$};
	\node[draw=none] at (.5,-.5) {$\dots$};
	\path (-1,-.3) -- node[named edge,left,pos=.75]{$\negatom{Y_1}$} (bc1);
	\path (-.2,-.3) -- node[named edge,pos=.75]{$\negatom{q_1}$} (bc2);
	\path (1,-.3) -- node[named edge,right,pos=.75]{$\posatom{X}/\qreg^+$} (bcn);
\end{scope}
\end{tikzpicture}
\end{adjustbox}
\caption{
	Grammar of nodes for proof-nets;
	atoms in the conclusions of a box are pairwise distinct.
}
\label{fig:grammar_nodes_pn}\label{fig:pn}
\end{figure}

\subparagraph{Quantum \pns.}

As before, we identify each name $X$ with the \rv\ $(X, \{x^{\true}, x^{\false}\})$ and each q-name $q_i$ with a distinct qubit.
With this implicit assumption, if $P$ is the \main\ of a box $\bbox$, then $\Nm P$ is either a random variable $\{X\}$ or a quantum register $\{q_1, \dots, q_k\}$.

\begin{definition}[Quantum \pns]\label{def:qbn}
A \definitive{\qpn} (\qbpn) is a pair $(\somepn, \Phi)$ where:
\begin{enumerate}
\item\label{item:qbn:1}
$\somepn$ is a proof-net such that all the atoms appearing in the outputs of distinct boxes are pairwise distinct, and all the negative q-atoms which are inputs of distinct boxes are pairwise distinct.
\item
$\someqfactset$ assigns to each box $\bbox$ of \main $P$ a \qCPT\ $\someqfactset(\bbox)$, for $\Nm{P}$ given the (classical and q-) names of the inputs of $\bbox$.
\end{enumerate}
\end{definition}

Thanks to point \ref{item:qbn:1}, we denote a box in a \qbpn $\somepn$ by its positive conclusion---\eg\ $\boxof{P}$ for the unique box of \main $P$. When $\Phi$ is irrelevant, we often write ``a \qbpn $\somepn$'' for ``a \qbpn $(\somepn,\Phi)$''.

\subparagraph{Reduction rules \& Normal forms.}

\Qpns are equipped with the standard reduction rules for \MLL proof-nets, presented in \Cref{fig:rewriting_pn}.
These rules define a binary relation $\Red$ on \qbpns, written $\somepn \Red \somepn'$ (read $\somepn$ reduces to $\somepn'$).
A \qbpn $\somepn$ is in \definitive{normal} form (or just normal) if there is no $\somepn'$ such that $\somepn \Red \somepn'$. 
Please notice that, because of boxes, a proof-net in normal form can still contain $\cut$-nodes---see \eg\ the \pn in \cref{ex:bell_pn}.

\begin{figure}
\begin{adjustbox}{}
\begin{tikzpicture}[baseline=(base)]
\coordinate (base) at (0,-.25);
\begin{scope}[every node/.style={draw=none}, every path/.style={draw=black}]
	\node (ax) at (0,0) {$\ax$};
	\node (cut) at (1.5,-.5) {$\cut$};
	\coordinate (c) at (-.75,-.5);
	\coordinate (m) at (.75,-.25);
	\coordinate (h) at (2.25,0);
	\path (ax) -| node[named edge,pos=.75,left]{$A\orth$} (c);
	\path (ax) -| node[named edge,pos=1]{$A$} (m);
	\path (cut) -| (m);
	\path (cut) -| node[named edge,pos=.75,right]{$A\orth$} (h);
\end{scope}
\end{tikzpicture}
$\quad\Red\quad$
\begin{tikzpicture}[baseline=(base)]
\coordinate (base) at (0,0);
\begin{scope}[every node/.style={draw=none}, every path/.style={draw=black}]
	\coordinate (c) at (0,-.5);
	\coordinate (h) at (0,.5);
	\path (c) -- node[named edge,pos=.5,right]{$A\orth$} (h);
\end{scope}
\end{tikzpicture}
\qquad\vrule\qquad
\begin{tikzpicture}[baseline=(base)]
\coordinate (base) at (0,0);
\begin{scope}[every node/.style={draw=none}, every path/.style={draw=black}]
	\node (tens) at (-.75,0) {$\tens$};
	\coordinate (tens1) at (-1.25,.5);
	\coordinate (tens2) at (-.25,.5);
	\node (par) at (1.75,0) {$\parr$};
	\coordinate (par1) at (1.25,.5);
	\coordinate (par2) at (2.25,.5);
	\node (cut) at (.5,-.5) {$\cut$};
	\path[out=-90,in=180] (tens1) edge node[named edge,left]{$A$} (tens);
	\path[out=-90,in=0] (tens2) edge node[named edge,right]{$B$} (tens);
	\path (tens) |- node[named edge,left,pos=.3]{$A \tens B$} (cut);
	\path[out=-90,in=180] (par1) edge node[named edge,left]{$A\orth$} (par);
	\path[out=-90,in=0] (par2) edge node[named edge,right]{$B\orth$} (par);
	\path (par) |- node[named edge,right,pos=.3]{$A\orth \parr B\orth$} (cut);
\end{scope}
\end{tikzpicture}
$\quad\Red\quad$
\begin{tikzpicture}[baseline=(base)]
\coordinate (base) at (0,0);
\begin{scope}[every node/.style={draw=none}, every path/.style={draw=black}]
	\coordinate (tens1) at (-1.25,.5);
	\coordinate (tens2) at (-.25,.5);
	\coordinate (par1) at (1.25,.5);
	\coordinate (par2) at (2.25,.5);
	\node (cut1) at (0,0) {$\cut$};
	\node (cut2) at (1,-.5) {$\cut$};
	\path[out=-90,in=180] (tens1) |- node[named edge,left,pos=.3]{$A$} (cut1);
	\path[out=-90,in=0] (tens2) |- node[named edge,left,pos=.4]{$B$} (cut2);
	\path[out=-90,in=180] (par1) |- node[named edge,right,pos=.3]{$A\orth$} (cut1);
	\path[out=-90,in=0] (par2) |- node[named edge,right,pos=.4]{$B\orth$} (cut2);
\end{scope}
\end{tikzpicture}
\end{adjustbox}
\caption{Reduction rules of proof-nets.}
\label{fig:rewriting_pn}
\end{figure}

\begin{proposition}
The reduction $\Red$ preserves both the correctness and the conclusions of a proof-net, is terminating and is confluent.
\end{proposition}

\subsection{Quantum Bayesian Networks and Quantum Proof-Nets}\label{sec:QBN_qpn}

We now relate the two formalisms, quantum Bayesian networks and quantum proof-nets.

\subparagraph{Polarized proof-nets as directed graphs.}

Polarities embed a notion of orientation (\Cref{rem:polarity}), allowing to read a typed graph as a DAG.
Restricting our attention to proof-nets whose edges are labelled only by polarized formulas, known as \textdef{polarized proof-nets}~\cite{Laurent03}, to each edge is associated a \emph{direction}: upwards for edges labelled by a negative formula, downwards for edges labelled by a positive formula.
The following is well-known for polarized proof-nets.

\begin{lemma}[Polarized correctness~\cite{phdlaurent}]\label{lem:pol_correct}
Let $\somepn$ be a typed graph (\Cref{def:pn}) whose edges are labelled only by polarized formulas. 
The graph $\somepn$ is \emph{correct} if and only if orienting the edges of $\somepn$ according to their polarity yields a DAG.
\end{lemma}

By polarized correctness, every polarized \pn\ $\somepn$ can be seen as a DAG where edges are oriented according to their polarity. 
This induces a DAG on the boxes $\Boxes\somepn=\{\boxof{P_1}, \dots, \boxof{P_n}\}$ of $\somepn$, hence a DAG $\G_\somepn$ on $\Nm{P_1}, \dots, \Nm{P_n}$: we have a directed edge $\Nm{P_i} \rightarrow \Nm{P_j}$ in $\G_\somepn$ whenever there is a directed path from $\bbox^{P_i}$ to $\bbox^{P_j}$ in $\somepn$.

\subparagraph{From \qpns to \qqbn's.}

We call \definitive{\good} a \qpn $\somepn \der \Delta$ where all the atoms appearing (as sub-formulas) in $\Delta$ are \emph{positive and classical}.

\begin{lemma}
To each \emph{polarized} \qbpn $(\somepn,\Phi_\somepn)$ of boxes $\Boxes{\somepn}=\{\bbox^{P_1}, \dots, \bbox^{P_n}\}$ is associated a pair $(\G, \Phi_\G)$ with $\G$ the DAG induced by $\somepn$ on $\Nm{P_1}, \dots, \Nm{P_n}$ and $\Phi_\G(\Nm P)\defeq \Phi_\somepn(\bbox^P)$.
If $\somepn$ is \good, the pair $(\G, \Phi_\G)$ is a \QBN.
\end{lemma}
 
\begin{example}
The \pn\ in \cref{ex:bell_pn} is polarized and \good; its DAG is the one in \Cref{fig:Bell}.
\end{example}


\subparagraph{From \qqbn s to \qpns.}

A quantum Bayesian network can be encoded as a \qpn rather directly.
To a \qqbn\ $(\G, \Phi)$, it is straightforward to associate \emph{a typed graph} $\Net_\G$ as sketched in \cref{fig:Bell,ex:bell_pn}.
It is immediate that $\Net_\G$ is polarized and \good.
The only delicate point is checking $\Net_\G$ is correct, hence a proof-net: it follows from \cref{lem:pol_correct}, as $\G$ is a DAG.
The pair $(\Net_\G, \Phi)$ clearly satisfies all conditions in \cref{def:qbn}, making it a \qbpn.

\subparagraph{Every pn has a polarized core.}

The encoding of \QBNs into polarized \good\ \qpns is sound and complete.
In fact, we can associate a \QBN\ to any \good \qpn (first reducing it to its normal form).

\begin{proposition}[Polarized core of a Normal Form]
Every \emph{normal} \pn $\somepn \der F$ consists of a polarized \pn $\pol{\somepn}$ on top of the formula tree of $F$.
\end{proposition}
This generalizes in the obvious way to $\somepn \der \Delta$.
For example, consider the proof-nets of \cref{fig:modularity_pn}: each has a formula tree drawn with \textcolor{red}{dashed} edges below its polarized core with \textcolor{black}{solid} edges.

Putting everything together, we can associate a \qqbn\ to (the normal form of) any \good \qbpn.
In the next section, we prove they have the \emph{same semantics} (\Cref{th:sem_coherence}).

\subsection{Compositionality and Modularity}
\label{sec:PNcompositionality}

We have now built all the ingredients to satisfy the desiderata outlined in the introduction.

\subparagraph{Denotational Semantics.}

We denote by $\Nm{\somepn}$ the set of all names appearing in $\somepn$.
\begin{definition}[Semantics of a \qpn]
\label{def:sem}
Let $(\somepn,\Phi)$ be a \qbpn.
Its semantics is:
\begin{equation*}
\den{\somepn} \defeq \sum_{\Nm{\somepn}\difset\Nm{\Delta}} \left(\bigodot_{\bbox\in \Boxes{\somepn}} \Phi(\bbox)\right)
\quad\text{where $\Delta$ is the conclusions of $\somepn$.}
\end{equation*}
\end{definition}

\begin{remark}\label{ex:var_ax}
If $\somepn$ is such that $\Boxes{\somepn}=\emptyset$, then $\den{\somepn}$ is the trivial factor $\ftone$.
\end{remark}

The invariance of the semantics via reduction is immediate, because both the boxes and the conclusions are invariants of the reduction.

\begin{proposition}[Invariance]
Let $\somepn$ be a \qpn.
\begin{enumerate}
\item
If $\somepn \Red \somepn'$ then $\somepn'$ is a \qpn with $\den{\somepn}=\den{\somepn'}$.
\item
Let $\somepn'$ be the normal form of $\somepn$, then $\den{\somepn} = \den{\somepn'}= \den {\pol{\somepn'}}$.
\end{enumerate}
\end{proposition}

The crucial question is \emph{compositionality}.
We prove that the semantics of a \qbpn $\somepn$ can be defined compositionally for any decomposition of $\somepn$ in sub-graphs.

\begin{theorem}[Compositionality]
\label{th:compositional}
Let $(\somepn,\Phi)$ be a \qpn, decomposed into two sub-graphs $\somepn_1$ and $\somepn_2$ with the obvious \qCPTs assignments given by restriction of $\Phi$.
Then:
\begin{equation*}
\den{\somepn} = \sum_{(\Nm{\Delta_1} \cup \Nm{\Delta_2})\difset\Nm{\Delta}} \big(\den{\somepn_1} \odot \den{\somepn_2}\big)
\end{equation*}
where $\Delta$ (\resp\ $\Delta_1$, $\Delta_2$) is the conclusions of $\somepn$ (\resp\ $\somepn_1$, $\somepn_2$).
\end{theorem}
The proof is similar to~\cite{BPN}, adapting to a graph setting the argument from~\cite{popl24}.

\subparagraph{Soundness and Completeness.}

In \Cref{sec:QBN_qpn}, we proved that the formalism of quantum proof-nets is sound and complete \wrt\ QBNs.
The semantics of a \qbpn\ and of the associated \qqbn\ are the same, essentially by definition.

\begin{theorem}[Completeness and Soundness with respect to QBNs]
\label{th:sem_coherence}
\hfill
\begin{itemize}
\item
To every \qqbn\ $\bn$ is associated a \good \qbpn $\somepn_\bn$ with $\sem{\bn} = \den{\somepn_{\bn}}$.
\item
To every \emph{\good} \qbpn\ $\somepn\der \Delta$ is associated a \qqbn\ $\bn_{\somepn}$ with $\den{\somepn} = \sum_{\Nm{\somepn} \difset \Nm{\Delta}}\sem {\bn_\somepn}$.
\end{itemize}
\end{theorem}

%

\subparagraph{On Types and Modularity.}

The typing discipline enables modularity: two \qbpns of dual type $F$ and $F\orth$ are \emph{guaranteed} to compose well.
More generally, composition is defined by the standard $\cut$-rule of sequent calculus---we say that $\R_1,\R_2$ below have \definitive{compatibles types}:
\begin{equation*}
\infer[\cut]{\somepn\der \Gamma,\Delta}{\somepn_1 \der \Gamma, F & \somepn_2 \der \Delta, F\orth}
\end{equation*}
where $\somepn$ is obtained by connecting the conclusions (labelled by) $F$ and $F\orth$ with a $\cut$-node.

\begin{proposition}
\label{th:types}
If the \qbpns $\somepn_1$ and $\somepn_2$ have compatible types, their composition is a \qbpn.
Plugging together the DAGs respectively associated to $\somepn_1$ and $\somepn_2$ produces a DAG.
\end{proposition}

\cref{fig:modularity_pn} revisits the motivational example of \Cref{fig:modular} in our typed setting. \cref{fig:modularity_pn_bad} presents a non-example of typing for $\N_3$ in \Cref{fig:modular}: this graph is not a proof-net due to the highlighted cycle.
Types enable composition, by acting as an interface which encodes (via $\otimes$ and $\parr$) the DAG's structure.

\begin{figure}
\begin{minipage}[b]{0.60\linewidth}
\begin{adjustbox}{}
\begin{tikzpicture}[baseline=(base)]
\coordinate (base) at (0,-2.6);
\node at (1.5,0) {$\somepn_0$};
\begin{scope}[every node/.style={rectangle,draw=none}, every path/.style={draw=black}]
	\node[draw=black,minimum width=3em] (bB) at (0,0) {$B$};
	\node[draw=black,minimum width=3em] (bC) at (3,0) {$C$};

	\node (cB) at (1.5,-.75) {$\cut$};
	\path (0.3,-.25) |- node[named edge,pos=.5,above right]{$B^+$} (cB) -| node[named edge,pos=.5,above left]{$B^-$} (2.6,-.25);
	\path (-.2,-.25) -- node[named edge,pos=.75,left]{$A^-$} ++(0,-.5);
	\path (3.3,-.25) -- node[named edge,pos=.75,right]{$C^+$} ++(0,-.5);
\end{scope}
\begin{scope}[every node/.style={rectangle,draw=none}, every path/.style={draw=red,dashed}]
	\node (parr) at (1.5,-1.5) {\color{red}$\parr$};
	\path[out=-90,in=180] (-.2,-.8) edge (parr);
	\path[out=-90,in=0] (3.3,-.8) edge (parr);
	\path (parr) -- node[named edge,pos=.75]{\color{red}$A^+\lollipop C^+$} ++(0,-.5);
\end{scope}
\end{tikzpicture}
\qquad
\begin{tikzpicture}
\node at (6.5,-.5) {$\somepn_1$};
\begin{scope}[every node/.style={rectangle,draw=none}, every path/.style={draw=black}]
	\node[draw=black,minimum width=3em] (bE) at (5,-.5) {$E$};
	\node[draw=black,minimum width=3em] (bD) at (10,-.5) {$D$};
	\node[draw=black,minimum width=3em] (bA) at (8,-.5) {$A$};

	\node (cE) at (6.5,-1.5) {$\cut$};
	\node (pE) at (8.9,-1) {$\cc$};
	\path (5.3,-.75) |- node[named edge,pos=.5,above right]{$E^+$} (cE) -| node[named edge,pos=.2,above left]{$E^-$} (pE);
	\path[out=180,in=-90] (pE) edge (7.8,-.75);
	\path[out=0,in=-90] (pE) edge (9.6,-.75);
	\path (9.8,-.75) -- node[named edge,pos=1,left]{$C^-$} ++(0,-1.25);
	\path (10.3,-.75) -- node[named edge,pos=.9]{$D^+$} ++(0,-2.25);
	\path (8.3,-.75) -- node[named edge,pos=1]{$A^+$} ++(0,-1.25);
\end{scope}
\begin{scope}[every node/.style={rectangle,draw=none}, every path/.style={draw=red,dashed}]
	\node (parr) at (9,-2.5) {\color{red}$\tens$};
	\path[out=-90,in=180] (8.3,-2) edge (parr);
	\path[out=-90,in=0] (9.8,-2) edge (parr);
	\path (parr) -- node[named edge,pos=.75,left]{\color{red}$(A^+\lollipop C^+)\orth$} ++(0,-.5);
\end{scope}
\end{tikzpicture}
\end{adjustbox}
\caption{Modularity through typing.}
\label{fig:modularity_pn}
\end{minipage}
\hfill
\begin{minipage}[b]{0.38\linewidth}
\begin{adjustbox}{}
\begin{tikzpicture}
\node at (6.5,-6.5) {$\somepn_2$};
\begin{scope}[every node/.style={rectangle,draw=none}, every path/.style={draw=black}]
	\node[draw=black,minimum width=3em] (bE) at (5,-6.5) {$E$};
	\node[draw=black,minimum width=3em] (bD) at (8,-6.5) {$D$};
	\node[draw=black,minimum width=3em] (bA) at (10,-6.5) {$A$};

	\node (cE) at (6.5,-7.5) {$\cut$};
	\node (pE) at (8.9,-7) {$\cc$};
	\node (cD) at (9,-8) {$\cut$};
	\node (pD) at (10,-7.5) {$\cc$};
	\node (axD) at (11,-7) {$\ax$};
\end{scope}
\begin{scope}[every node/.style={rectangle,draw=none}, every path/.style={draw=red,dashed}]
	\node (parr) at (9,-8.5) {\color{red}$\tens$};
\end{scope}
\begin{scope}[every path/.style={thick,draw=yellow,-,line width=4pt}]
	\path[out=180,in=-90] (pD) edge (9.8,-6.75);
	\path (8.3,-6.75) |- (cD) -| (pD);
	\path (7.6,-6.75) -- ++(0,-1.25);
	\path (10.3,-6.75) -- ++(0,-1.25);
	\path[out=-90,in=180] (7.6,-8) edge (parr);
	\path[out=-90,in=0] (10.3,-8) edge (parr);
\end{scope}
\begin{scope}[every node/.style={rectangle,draw=none}, every path/.style={draw=black}]
	\path (5.3,-6.75) |- node[named edge,pos=.5,above right]{$E^+$} (cE) -| node[named edge,pos=.2,above left]{$E^-$} (pE);
	\path[out=180,in=-90] (pE) edge (7.8,-6.75);
	\path[out=0,in=-90] (pE) edge (9.6,-6.75);
	\path (8.3,-6.75) |- node[named edge,pos=.5,above right]{$D^+$} (cD) -| node[named edge,pos=.5,above left]{$D^-$} (pD);
	\path[out=180,in=-90] (pD) edge (9.8,-6.75);
	\path[out=0,in=-90] (pD) edge (10.5,-7);
	\path (axD) -| (10.5,-7);
	\path (axD) -| node[named edge,pos=.95]{$D^+$} (11.5,-9);
	\path (7.6,-6.75) -- node[named edge,pos=1,left]{$C^-$} ++(0,-1.25);
	\path (10.3,-6.75) -- node[named edge,pos=1,right]{$A^+$} ++(0,-1.25);
\end{scope}
\begin{scope}[every node/.style={rectangle,draw=none}, every path/.style={draw=red,dashed}]
	\path[out=-90,in=180] (7.6,-8) edge (parr);
	\path[out=-90,in=0] (10.3,-8) edge (parr);
	\path (parr) -- node[named edge,pos=.75,left]{\color{red}$(A^+\lollipop C^+)\orth$} ++(0,-.5);
\end{scope}
\end{tikzpicture}
\end{adjustbox}
\caption{Non-example of proof-net.}
\label{fig:modularity_pn_bad}
\end{minipage}
\end{figure}

%% file: 05_Conclusions.tex
\section{Conclusion and Related Work}

This paper brings compositional principles and a typing discipline in the setting of \QBNs, addressing the lack of compositionality and modularity in the original setting with methods from denotational semantics and proof theory.
Our framework is fully compatible with classical Bayesian networks and Bayesian inference.
The crucial notion is that of \emph{quantum factors}, that directly generalize factors from the theory of (classical) Bayesian networks and satisfy the same key properties for product and sum.
\begin{center}
\begin{tabular}{|c|c|}
\hline
Classical factor over \rvs $\bX$ & Quantum factor over \rvs $\bX$ and the quantum register $\qreg$
\\
\hline
$\phi:\setofrv{\bX} \to \nnReal$	& $\someqfact: \setofrv{\bX} \to \Pos{\HH_\qreg}$
\\
\hline
\end{tabular}
\end{center}
Their development is highly non-trivial, and is our main technical contribution.
A major challenge has been to conciliate the behaviour of classical factors with quantum no-cloning.
Our product of \qfactor s shares the value of classical variables (as factors do) and coincides with tensor contraction for quantum registers, being fully compatible with both settings.

The framework enables computing the probability distribution associated to a network by multiplying its \qfactor s in any arbitrary order (not necessarily in a top-down fashion as in~\cite{GBN2014}), and by summing out irrelevant variables on sub-components.

Furthermore, we propose a typed formalism for \QBN s---based on proof-nets---that guarantees \emph{modularity} through a proof-theoretic approach.
The \emph{exact} correspondence between closed proof-nets and QBNs is guaranteed by soundness and completeness results.
Finally, it is worth noting that \qfactor s are a model of Multiplicative Linear Logic.

\subparagraph{Related work.}

Our semantics is strongly inspired by Selinger's~\cite{Selinger2004}---the graph of the function defining a \emph{quantum factor} can be seen as an \emph{indexed tuple of matrices}, closely related to the tuples of matrices in~\cite{Selinger2004}.
Another notable denotational model built over density operators and quantum operations is~\cite{HasuoH17}.
Our typed graphical formalism based on linear logic exploits and expands a recent line of work relating Bayesian networks with \pns~\cite{EhrhardFP23,BPN} and with the information flow on type derivations~\cite{popl24}. 

As discussed in the introduction, the literature on quantum causal models (see \eg~\cite{BarrettLorenzOreshkov2019} and references therein) is vast and multifaceted.
We favor the \emph{inference} perspective propounded in~\cite{LeiferSpekkens2013}, a work which is preliminary to QBNs~\cite{GBN2014}.
While we only cite~\cite{LeiferSpekkens2013}, a substantial body of work studies quantum extension of Bayes' rule and Bayesian inference in quantum mechanics and quantum foundations.
This is however beyond the scope of our work.

Our goal is a \emph{conservative} extension of Bayesian networks to the quantum setting.
An early definition of quantum Bayesian networks---quite different from the one we consider---is given by Tucci~\cite{Tucci1995}, where amplitudes replace probabilities.
We focus on the work by Henson, Lal, and Pusey~\cite{GBN2014}, that introduces a notion of \emph{generalized Bayesian networks} including the quantum case; recent work in the same line is~\cite{KhannaAnsanelliPuseyWolfe2024}.
A closely connected formalism is that of~\cite{Fritz2012}, which focuses on (quantum) non-locality in scenarios with several quantum sources.
The results in~\cite{GBN2014} exploit also the framework of \emph{quantum networks} developed in~\cite{10.1103/PhysRevA.80.022339}.
It should be noted that \emph{quantum networks} and \emph{quantum Bayesian networks}, despite their similar names, are distinct frameworks with different purposes.
In particular, quantum networks cannot encode BNs: a classical operational-probabilistic theory in~\cite{10.1103/PhysRevA.80.022339} is not a BN---see the discussion in~\cite[Section~3.1 and Example~10]{GBN2014}.

Taking a broader perspective beyond \QBNs, we briefly mention some active research areas which share related aims. 
The development of a higher-order typed framework for quantum processes~\cite{10.1098/rspa.2018.0706,10.22331/q-2026-01-21-1978} is a timely and vibrant research line at the frontier between quantum physics and higher-order computation, exploring quantum operations on quantum operations.
Another very active line of research, concerned with compositionality and higher-order structures in both quantum and probabilistic settings, is that of string diagrams, adopting a categorical approach---see \eg~\cite{10.23638/LMCS-15(3:15)2019,10.4230/LIPIcs.MFCS.2022.80, simmons2024completelogiccausalconsistency}.
This fruitful line has a perspective different from ours.
We noted already that BNs have a dual nature, as tools for \emph{efficient probabilistic inference} and as \emph{causal models}.
While we favor the former, the line of developments based on monoidal category theory focuses on causality and is unconcerned with the computational cost of probabilistic reasoning.
As observed in~\cite[Example~9.7]{popl24}, the cost of \emph{actually} computing the semantics of a BN explodes when taking this approach, because a central notion is a product $\otimes$ which behaves like the tensor product of matrices.
We give an example in \cref{app:cost}. 

Our proof-net syntax provides an \emph{exact} characterization of QBNs, setting our work apart from other frameworks.
In the last two decades, several papers have developed variants of proof-nets accounting for quantum processes~\cite{DuncanThesis,AbramskyDuncan2009GeneralisedProofNets,esop2014,LagoFVY17,simmons2024completelogiccausalconsistency}; here we briefly discuss~\cite{simmons2024completelogiccausalconsistency} which introduces proof-nets as a tool to analyze \emph{causal} structures.
A key ingredient, which is also key in~\cite{popl24} and in our own treatment, is the (polarized) flow of information.
The framework in~\cite{simmons2024completelogiccausalconsistency} focuses on soundness and completeness (completeness being non‑trivial) \wrt\ causal string diagrams, building on Retoré’s proof-nets: causality is encoded via a non‑commutative tensor.
In contrast, we present a class of proof‑nets that is sound and complete (soundness being non‑trivial) \wrt\ QBNs.
When all data are classical, we recover BNs and~\cite{BPN}, where standard inference algorithms (\eg\ Variable Elimination, Message Passing) are available.

%% file: 99_Appendix.tex
\appendix

\section{Factors product vs Tensor product}
\label{app:cost}\label{app:section2}


We stress how much the product of factors in the theory of Bayesian networks (noted $\odot$ in our paper) differs from the tensor product of matrices. 
Consider the following BN (corresponding to the term in~\cite[Example 9.7]{popl24}), whose semantics is $\Pr(X_1, X_2,Y_1,Y_2,Y_3)$, a table with $2^5$ entries.

\begin{adjustbox}{width=0.7\linewidth, center=\linewidth}
\begin{tikzpicture}
\begin{genscope}
\node (Y1) at (2,0) {$Y_1$};
\node (Y2) at (6,0) {$Y_2$};
\node (Y3) at (10,0) {$Y_3$};
\node (X1) at (4,-2) {$X_1$};
\node (X2) at (8,-2) {$X_2$};
\draw[->] (Y1) edge (X1);
\draw[->] (Y1) edge (X2);
\draw[->] (Y2) edge (X1);
\draw[->] (Y2) edge (X2);
\draw[->] (Y3) edge (X1);
\draw[->] (Y3) edge (X2);
\end{genscope}
\end{tikzpicture}
\end{adjustbox}
The semantics $\den{X_1}$ of the node $X_1$ is $\Pr(X_1 \mid Y_1,Y_2,Y_3)$ (a factor $\phi_1(x_1,y_1,y_2,y_3)$), which can be seen as a stochastic matrix of size $2^4$.
Similarly, the semantics $\den{X_2}$ of the node $X_2$ is $\Pr(X_2 \mid Y_1,Y_2,Y_3)$ and can be seen as a stochastic matrix with $2^4$ entries.

By using the respective definitions, it is immediate to check that:
\begin{itemize}
\item
The tensor product $\den {X_1} \otimes \den {X_2} $ yields a table with $2^8$ entries.
\item
The factors product $\den {X_1} \odot \den {X_2} $ yields a table with $2^5$ entries (because non compatible instantiations are never computed).
\end{itemize}

\subparagraph{On the computational cost of the semantical interpretation: $\odot$ vs $\otimes$.}

\emph{Inference} algorithms on BNs aim at computing the semantics of the model through intermediate, partial computations, \emph{without ever computing the full joint probability} (as unfeasible in practice).

The approaches based on monoidal categories typically focus on \emph{causality}, not on the cost of inference.
A central role in the semantics is here played by a product $\otimes$ behaving as the tensor product of matrices.
Computing the semantics of $n$ binary random variables easily leads to intermediate computations of size much larger than $2^n$, the size of the full joint distribution.
Looking again at the BN above, consider the sub-net containing only the nodes $X_1$ and $X_2$.
Its factor-based semantics is $\den{X_1} \odot \den{X_2}$.
In a categorical setting, one computes $\den{X_1} \otimes \den{X_2}$, which is \emph{larger} than the full joint probability distribution (the worst-case scenario in BN's inference).
The point is that \emph{inference is not the focus} there.


\section{Proof of Proposition~\texorpdfstring{\ref{lem:qfactors_op}}{18}}
\label{app:section3}

\qfactors*
\begin{proof}
We only prove closeness, the rest being simple computations for each claimed equation.
\begin{itemize}
\item
Consider $\someqfact_1\compfact \someqfact_2$ with $\someqfact_i$ a \qfactor\ over $(\bX_i,\qreg_i)$.
We will use the completely positive map $\qcup[\L(\HH_{\qreg})]:\L(\HH_{\qreg})\tens\L(\HH_{\qreg})\to\C$ defined on the canonical orthonormal basis by
\begin{equation*}
\qcup[\L(\HH_{\qreg})](\be \tens \be') \defeq 1 \text{ if } \be=\be' \text{ and } 0 \text{ otherwise}
\end{equation*}
and extended to the general case by linearity.
Fix $\bx\in\setofrv{\bX_1\cup\bX_2}$, with $\bx_i$ its restriction to $\setofrv{\bX_i}$.
Then, $\someqfact_1(\bx_1)$ is a positive matrix in $\L(\HH_{\qreg_1})\cong\L(\HH_{\qreg_1\difset\qreg_2})\tens\L(\HH_{\qreg_1\cap\qreg_2})$ and $\someqfact_2(\bx_2)$ is a positive matrix in $\L(\HH_{\qreg_2})\cong\L(\HH_{\qreg_1\cap\qreg_2})\tens\L(\HH_{\qreg_2\difset\qreg_1})$.
We want positivity of $(\someqfact_1 \compfact \someqfact_2) (\bx) \in \L(\HH_{\qreg_1\symdif\qreg_2})\cong\L(\HH_{\qreg_1\difset\qreg_2})\tens\L(\HH_{\qreg_2\difset\qreg_1})$.
We remark that:
\begin{equation*}
(\someqfact_1 \compfact \someqfact_2) (\bx)
=
\left(\id[\L(\HH_{\qreg_1\difset\qreg_2})] \tens \qcup[\L(\HH_{\qreg_1\cap\qreg_2})] \tens \id[\L(\HH_{\qreg_2\difset\qreg_1})]\right)
\left(\someqfact_1(\bx_1) \tens \someqfact_2(\bx_2)\right)
\end{equation*}
As $\someqfact_1(\bx_1) \tens \someqfact_2(\bx_2)$ is positive and $\qcup$ is completely positive, $(\someqfact_1 \compfact \someqfact_2) (\bx)$ is positive.
\item
Consider a \qfactor\ $\someqfact$ over $(\bY\uplus\{X\},\qreg)$.
Fixing $\by\in\setofrv{\bY}$, as each $\someqfact(\by,x)$ is a positive matrix, it follows that their sum $(\sum_X \someqfact) (\by) = \sum_{x\in\setofrv{X}}\someqfact(\by,x)$ is positive too.
\item
Consider a \qfactor\ $\someqfact$ over $(\bX,\qreg\uplus\qreg')$.
For any $\bx\in\setofrv{\bX}$, $(\sum_{\qreg'} \someqfact) (\bx) = \ptr{\HH_{\qreg'}}\left(\someqfact(\bx)\right)$ is positive since $\someqfact(\bx)$ is positive and the (partial) trace is completely positive.
\qedhere
\end{itemize}
\end{proof}

\section{Around \QCPT s}
\label{app:equivalence_qbns}

Notice that our semantics of quantum Bayesian networks does not distinguish inputs and outputs of nodes, similarly to the factors semantics or to the relational models of Linear Logic.
For instance, the interpretation of a node with a quantum input $\qreg_1$, a quantum output $\qreg_2$ and a classical output $X$ is a function from $\setofrv{X}$ to positive matrices in $\Pos{\HH_{\qreg_1} \tens \HH_{\qreg_2}}$.

\subsection{Intuitions for \QCPT s}
\label{sec:QCPT}

We give some intuition for \cref{def:condition_qfactor_qbn}.
It is straightforward to associate a \qfactor\ to a quantum instrument (we do so in \cref{app:equivalence_qbns_proof}).
The reverse is not always possible, and \qCPTs are those \qfactor s corresponding to quantum instruments.

Consider a \qfactor\ $\someqfact$ over $(\{X\}\uplus\bY,\qreg)$.
We want it to be the interpretation of a node $X$ in a quantum Bayesian network: given a density matrix $\M\in\D(\HH_\qreg)$ and a value $\by\in\setofrv{\bY}$, we want a probability distribution on $X$.
A natural way to do so is by seeing $\M$ as a \qfactor\ over $(\emptyset,\qreg)$ (\ie\ the function sending the unique element of a singleton to $\M$): then, $\someqfact\compfact\M$ is a \qfactor\ over $(\{X\}\cup\bY,\emptyset)$, \ie\ each $(\someqfact\compfact\M)(x,\by)$ is in $\nnReal$.
Therefore, we say that $\someqfact$ is a \qCPT\ for $X$ given $(\bY, \qreg)$ if for all $\M\in\D(\HH_\qreg)$ and $\by\in\setofrv{\bY}$, $\sum_{x\in\setofrv{X}}(\someqfact\compfact\M)(x,\by) = 1$.

Similarly, consider a \qfactor\ $\someqfact$ over $(\bX,\qreg\uplus\qreg')$ that we want to associate to a node $\qreg$ of a network.
Hence, for every density matrix $\M\in\D(\HH_{\qreg'})$ and value $\bx\in\setofrv{\bX}$, we should get a density matrix in $\D(\HH_{\qreg})$.
As before, seeing $\M$ as a \qfactor\ over $(\emptyset,\qreg')$, let us look at the \qfactor\ $\someqfact\compfact\M$ over $(\bX,\qreg)$.
For $\someqfact$ to be a \qCPT\ for $\qreg$ given $(\bX,\qreg')$, we require that for all $\M\in\L(\HH_{\qreg'})$ and $\bx\in\setofrv{\bX}$, the positive matrix $(\someqfact\compfact\M)(\bx)$ has trace $1$.

\Cref{def:condition_qfactor_qbn} is simply a reformulation of the two conditions found above.

\begin{remark}
While complete positivity is relevant for instruments, for \qfactor s we only need to consider positive matrices.
This is because given $\Eop:\L(\C)\to\L(\HH)$---which is isomorphic to a matrix in $\L(\HH)$---complete positivity collapses to positivity.
\end{remark}

\subsection{Equivalence between Instrument- and \Qfactor-based \qqbn s}
\label{app:equivalence_qbns_proof}

We prove here a claim from \cref{sec:qbn}: our definition of quantum Bayesian networks with \qfactor s is equivalent to the one with quantum instruments from~\cite{GBN2014}, sketched in \cref{sec:qbn_instrument}.

\begin{proposition}
\label{lem:bij_2_QBNs}
There is a one-to-one correspondence between instrument- and \qfactor-based Bayesian networks, which preserves the semantics.
\end{proposition}

As explained in \cref{sec:choi_iso}, the two syntaxes consider the same underlying DAGs (up to inconsequential details).
What we have to do to obtain \Cref{lem:bij_2_QBNs} is:
\begin{itemize}
\item
Giving a bijection $\Psi$ between a \qCPT associated to a node in our presentation, and instruments associated to the same node in~\cite{GBN2014}.
\item
Proving that composition of instruments corresponds to product of their images by $\Psi$.
\end{itemize}

This bijection $\Psi$ is a simple adaptation of the Choi-Jamiolkowski isomorphism to a setting with classical inputs.
This builds upon the fact that \qCPTs are defined similarly to Choi states.
We first define $\Psi$, for nodes associated to a \rv\ and nodes associated to a register.

\begin{lemma}
\label{lem:qcpt_qi}
\hfill
\begin{itemize}
\item
Let $X$ be a \rv, $\bY$ a set of \rvs, and $\qreg$ a register.
There is a bijection between families of quantum instruments $\{\Eop_\by \mid \by\in\setofrv{\bY}\}$ where each $\Eop_\by$ is an instrument $\{\Eop_{x,\by} \mid x\in\setofrv{X}\}$ from $\L(\HH_{\qreg'})$ to $\C$, and \qCPTs on $X$ given $(\bY,\qreg')$.
\item
Let $\bX$ be a set of \rvs, and $\qreg$ and $\qreg'$ two registers.
There is a bijection between families of quantum operations $\Eop_\bx$ from $\L(\HH_{\qreg'})$ to $\L(\HH_\qreg)$, and \qCPTs on $\qreg$ given $(\bX,\qreg')$.
\end{itemize}
\end{lemma}
\begin{proof}
We generalize the definition of \qCPTs (\cref{def:condition_qfactor_qbn}) to get both results as particular cases.
A \qfactor\ $\someqfact$ over $(\bX\uplus\bY,\qreg\uplus\qreg')$ is a \textdef{\qCPT\ for $(\bX,\qreg)$ given $(\bY, \qreg')$} if $\forall \by\in\setofrv{\bY}$:
\begin{equation*}
\left(\sum_{\bX,\qreg} \someqfact\right)(\by)
=
\id[\L(\HH_{\qreg'})]
\end{equation*}

We give a bijection $\Psi$ (along with its inverse $\Psi^{-1}$) between families of quantum instruments $\{\Eop_\by \mid \by\in\setofrv{\bY}\}$, where each $\Eop_\by$ is a quantum instrument $\{\Eop_{\bx,\by} \mid \bx\in\setofrv{\bX}\}$ from $\L(\HH_{\qreg'})$ to $L(\HH_\qreg)$, and \qCPTs for $(\bX,\qreg)$ given $(\bY, \qreg')$.
To simplify notations, we write $(\be_i)_{i}$ (\resp\ $(\be'_k)_{k}$) the canonical orthonormal basis of $\L(\HH_{\qreg})$ (\resp\ $\L(\HH_{\qreg'})$).
We will use the following.
\begin{itemize}
\item
Call $\qcup[\L(\HH_{\qreg'})]:\L(\HH_{\qreg'})\tens\L(\HH_{\qreg'})\to\C$ the completely positive map defined by
\begin{equation*}
\qcup[\L(\HH_{\qreg'})](\be'_k \tens \be'_l) \defeq 1 \text{ if } k=l \text{ and } 0 \text{ otherwise}
\end{equation*}
and extended to the general case by linearity.
\item
Set $\qcap[\L(\HH_{\qreg'})]\defeq \sum_k \be'_k\tens \be'_k$, which is a positive matrix in $\L(\HH_{\qreg'})\tens\L(\HH_{\qreg'})$.
\end{itemize}

Let us now define $\Psi$ and $\Psi^{-1}$.
\begin{itemize}
\item
Given a family of quantum instruments $\{\Eop_\by \mid \by\in\setofrv{\bY}\}$ from $\L(\HH_{\qreg'})$ to $L(\HH_\qreg)$, we set $\Psi(\{\Eop_{\bx,\by} \mid \bx\in\setofrv{\bX}\})$ the \qfactor\ over $(\bX\uplus\bY,\qreg\uplus\qreg')$ defined by:
\begin{equation*}
\left(\Psi(\{\Eop_{\bx,\by} \mid \bx\in\setofrv{\bX}\})\right)(\bx,\by)
\defeq
(\id[\L(\HH_{\qreg'})] \tens \Eop_{\bx,\by})(\qcap[\L(\HH_{\qreg'})])
\end{equation*}

\item
Given a \qCPT\ $\someqfact$ for $(\bX,\qreg)$ given $(\bY, \qreg')$, we define for each $\by\in\setofrv{\bY}$ the following set of maps:
\begin{equation*}
\Psi^{-1}_{\by}(\someqfact)
\defeq
\{ \M \mapsto (\qcup[\L(\HH_{\qreg'})] \tens \id[\L(\HH_\qreg)]) (\M \tens \someqfact(\bx,\by)) \mid \bx\in\setofrv{\bX}\}
\end{equation*}
and pose $\Psi^{-1}(\someqfact) \defeq \{\Psi^{-1}_{\by}(\someqfact)\mid \by\in\setofrv{\by}\}$.
\end{itemize}
By similar computations to that involved for the Choi-Jamiolkowski isomorphism, $\Psi$ sends quantum instruments to \qCPTs, $\Psi^{-1}$ sends \qCPTs to quantum instruments, and $\Psi^{-1}$ is the inverse of $\Psi$ (the later computation using the yanking equations of $\qcup$ and $\qcap$).
\end{proof}

Showing the composition of instruments results in the product of their images by $\Psi$ is also a simple computation (using the yanking equations of $\qcup$ and $\qcap$).
This proves \cref{lem:bij_2_QBNs}.

\begin{remark}
\Cref{lem:sem_qbn_dist} is a corollary of \cref{lem:bij_2_QBNs}, using that the semantics of an instrument-based Bayesian network is a probability distribution.
Another corollary is that \qCPTs are closed under product, because instruments are closed by composition.
\end{remark}

